\documentclass[journal]{IEEEtran}
\usepackage{amsmath}
\usepackage{amssymb}
\usepackage{amsfonts}
\usepackage{graphicx}
\usepackage{epsfig}
\usepackage{subfigure}
\usepackage{psfrag}
\usepackage{xcolor}

\begin{document}

\title{Secrecy Wireless Information and Power Transfer with MISO Beamforming}

\author{Liang Liu, Rui Zhang,~\IEEEmembership{Member,~IEEE}, and Kee-Chaing Chua,~\IEEEmembership{Member,~IEEE}

\thanks{Manuscript received July 23, 2013, revised November 9, 2013 and January 17, 2014, accepted January 21, 2014. The associate editor
coordinating the review of this paper and approving it for
publication was Dr. De Maio Antonio. The paper has been presented in part at IEEE Global Communications Conference (Globecom), Atlanta, GA, December, 2013. The
work was supported in part by the National University of Singapore
under Research Grant R-263-000-679-133.}

\thanks{The authors are with the
Department of Electrical and Computer Engineering, National
University of Singapore (e-mail:liu\_liang@nus.edu.sg; elezhang@nus.edu.sg; eleckc@nus.edu.sg). R. Zhang is also
with the Institute for Infocomm Research, A*STAR, Singapore.}
}

\maketitle

\begin{abstract}

The dual use of radio signal for simultaneous wireless information and power transfer (SWIPT) has recently drawn significant attention. To meet the practical requirement that the energy receiver (ER) operates with significantly higher received power as compared to the conventional information receiver (IR), ERs need to be deployed in more proximity to the transmitter than IRs in the SWIPT system. However, due to the broadcast nature of wireless channels, one critical issue arises that the messages sent to IRs can be eavesdropped by ERs, which possess better channels from the transmitter. In this paper, we address this new physical-layer security problem in a multiuser multiple-input single-output (MISO) SWIPT system where one multi-antenna transmitter sends information and energy simultaneously to an IR and multiple ERs, each with one single antenna. Two problems are investigated with different practical aims: the first problem maximizes the secrecy rate for the IR subject to individual harvested energy constraints of ERs, while the second problem maximizes the weighted sum-energy transferred to ERs subject to a secrecy rate constraint for IR. We solve these two non-convex problems optimally by a general two-stage procedure. First, by fixing the signal-to-interference-plus-noise ratio (SINR) target for ERs (in the first problem) or IR (in the second problem), we obtain the optimal transmit beamforming and power allocation solution by applying the technique of semidefinite relaxation (SDR). Then, each of the two problems is solved by a one-dimension search over the optimal SINR target for ERs or IR. Furthermore, for each problem, suboptimal solutions of lower complexity are proposed in which information and energy beamforming vectors are separately designed from their power allocation. Simulation results are provided to compare the performances of proposed optimal and suboptimal designs in terms of (secrecy) rate and energy transmission trade-off between the IR and ERs.

\end{abstract}


\setlength{\baselineskip}{1.0\baselineskip}
\newtheorem{definition}{\underline{Definition}}[section]
\newtheorem{fact}{Fact}
\newtheorem{assumption}{Assumption}
\newtheorem{theorem}{\underline{Theorem}}[section]
\newtheorem{lemma}{\underline{Lemma}}[section]
\newtheorem{corollary}{\underline{Corollary}}[section]
\newtheorem{proposition}{\underline{Proposition}}[section]
\newtheorem{example}{\underline{Example}}[section]
\newtheorem{remark}{\underline{Remark}}[section]
\newtheorem{algorithm}{\underline{Algorithm}}[section]
\newcommand{\mv}[1]{\mbox{\boldmath{$ #1 $}}}

\section{Introduction}\label{eqn:Introduction}

\begin{figure}
\begin{center}
 \scalebox{0.45}{\includegraphics*{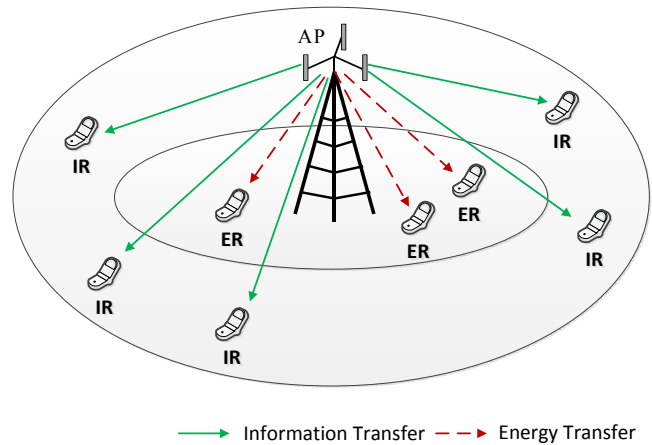}}
 \end{center}
\caption{A simultaneous wireless information and power transfer (SWIPT) system with ``near'' energy receivers (ERs) and ``far'' information receivers (IRs) from the AP.} \label{fig1}\vspace{-20pt}
\end{figure}

\PARstart{R}ecently, there has been an upsurge of interest in radio signal enabled simultaneous wireless information and power transfer (SWIPT) \cite{Sahai10}--\cite{Clerckx13}. A typical SWIPT system of practical interest is shown in Fig. \ref{fig1}, where a fixed access point (AP) with constant power supply broadcasts wireless signal to a set of distributed user terminals (UTs), among which some intend to decode information from the received signal, referred to as information receivers (IRs), while the others are interested in harvesting the signal energy, thus called energy receivers (ERs) \cite{Rui11}, \cite{Rui13}. To meet the practical requirement that the IR and ER typically operate with very different power sensitivity (e.g., $-60$dBm for IR versus $-10$dBm for ER), a ``near-far'' or receiver-location based scheduling for information and energy transmissions has been proposed in \cite{Rui11}, \cite{Rui13}, where ERs need to be deployed in more proximity to the AP than IRs for receiving higher signal power due to distance-dependent attenuation.

However, the receiver-location based transmission scheduling for SWIPT gives rise to a new information security issue since ERs, which are closer to the transmitter and thus have better channels than IRs, can more easily eavesdrop the information sent to IRs. In order to advance existing secrecy wireless communication network to a future one with hybrid information and energy transmission, it is thus an essential task to overcome this challenging security problem in SWIPT. Therefore, in addition to meeting the energy harvesting requirements of the ERs, the SWIPT system should be optimally designed to guarantee that the information is delivered securely to each IR even in the presence of possible eavesdropping by any of the ERs.

Motivated by the aforementioned problem, in this paper we propose the use of multiple antennas at the AP to achieve secret information transmission to the IRs and yet guarantee the target amount of energy simultaneously transferred to the ERs, by optimally designing the beamforming vectors and their power allocation at the transmitter. For the purpose of exposition, we consider a multiple-input single-output (MISO) SWIPT system consisting of one multi-antenna transmitter, one single-antenna IR, and $K\geq 1$ single-antenna ERs, as shown in Fig. \ref{fig2}. Two secrecy beamforming design problems, denoted by (P1) and (P2), are investigated for the MISO SWIPT system based on different performance requirements of IRs and ERs in practice. In (P1), we study the joint information and energy transmit beamforming design for the scenario where ERs have stringent energy harvesting requirements by maximizing the secrecy rate of the IR subject to individual harvested energy constraints of ERs. In contrast, for (P2), we investigate the case where the IR has a fixed-rate transmission and thus the weighted sum-energy transferred to ERs is maximized subject to a secrecy rate constraint for the IR. Both problems (P1) and (P2) are shown to be non-convex in general; however, we solve them globally optimally by reformulating each problem into two subproblems as follows. First, by fixing the signal-to-interference-plus-noise ratio (SINR) target at all ERs in (P1) or at the IR in (P2), we obtain the optimal beamforming and power allocation solutions by applying the technique of semidefinite relaxation (SDR) \cite{Luo10}. We show that there always exists a rank-one optimal transmit covariance solution for the IR in both (P1) and (P2), i.e., transmit beamforming is optimal for the IR; furthermore, if the SDR results in a higher-rank solution, we propose one efficient algorithm to construct an equivalent rank-one optimal solution. We also present the condition under which SDR always yields a rank-one optimal transmit covariance solution for the IR in (P1) or (P2), and show that this condition is practically satisfied in our studied MISO SWIPT system. Next, we implement a one-dimension search for the optimal SINR target of ERs in (P1) or of the IR in (P2) to optimally solve these two problems. Finally, we present two suboptimal solutions of lower complexity for each of the two studied problems, in which beamforming vectors are separately designed from their power allocation, and compare their performances to the optimal solutions in terms of achievable (secrecy) rate-energy transmission trade-off in MISO SWIPT systems.

\begin{figure}
\begin{center}
 \scalebox{0.5}{\includegraphics*{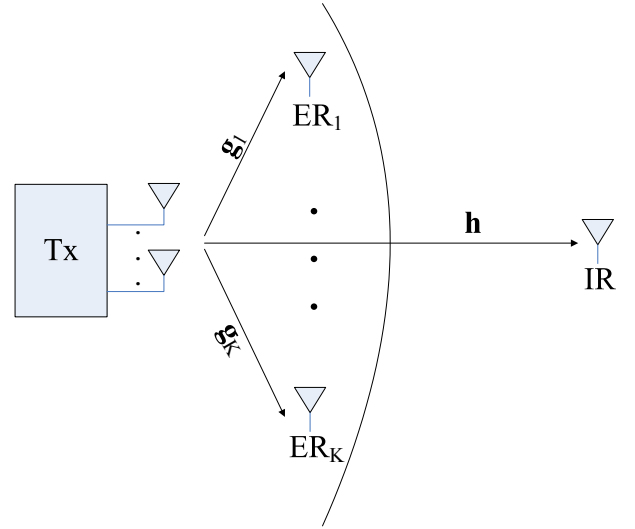}}
 \end{center}
\caption{A MISO SWIPT system with $K$ ``near'' ERs and one ``far'' IR.} \label{fig2}\vspace{-20pt}
\end{figure}

It is worth noting that in \cite{Rui13}, a similar MISO SWIPT system as in this paper has been studied, but under the assumption of independent user channels and without the secret information transmission constraint. It was shown in \cite{Rui13} that to maximize the weighted sum-energy transferred to ERs while meeting individual SINR constraints at IRs, the optimal strategy is to adjust information beams only without the use of any energy beam. However, in this paper we will show that with the newly introduced secrecy rate constraint at the IR, energy beams are in general needed in the optimal solution for the SWIPT system with arbitrary user channels. The reason is that energy beams in our new setup play an additional important role of artificial noise (AN) \cite{Goel06} to facilitate the secret information transmission to the IR by deliberately interfering with ERs that may eavesdrop the IR's message. It is also noted that AN-aided secrecy communication has been extensively studied in the literature (see e.g. \cite{Goel06}--\cite{Ma11} and the references therein), where a fraction of the transmit power is allocated to send artificially generated noise signals to reduce the amount of information that can be decoded by the eavesdroppers. Since in practice eavesdroppers' channels are in general unknown at the transmitter \cite{Rui09}, an isotropic approach was proposed in \cite{Goel06} where the power of AN is uniformly distributed in the null space of the legitimate receiver's channel, and the performance of this practical approach has been shown to be nearly optimal at the high signal-to-noise ratio (SNR) regime \cite{Khisti07}. With imperfect knowledge of eavesdroppers' channels at the transmitter, various AN-aided secrecy transmission schemes have also been proposed for different channel setups \cite{Swindlehurst12}--\cite{Mckay10}. Furthermore, the MISO beamforming design problem for the AN-aided secrecy transmission under the assumption that eavesdroppers' channels are perfectly known at the transmitter has been studied in e.g., \cite{MaKent11} and \cite{Ma11}. Note that this assumption may not be valid if the eavesdroppers are passive devices in practical systems. However, in the SWIPT system of our interest, since ERs need to assist the transmitter in obtaining their channel knowledge to design transmit beamforming to satisfy their individual energy requirements, it is more practically reasonable to assume that ERs' channels are known at the transmitter as in this paper.

It is worth pointing out further that besides the physical layer security approach, cryptography technique that is implemented in the higher protocol layer is an alternative solution for achieving secrecy information transmission (see e.g. \cite{Zhang07} and the references therein). In this paper, we focus our study on the physical layer security in SWIPT due to the following reasons. Firstly, cryptography requires a common private key shared by the transmitter and the legitimate receiver, which needs an additional secure communication channel for the key exchange. Secondly, the dual use of AN for both secrecy information and power transfer is an interesting new idea proposed in this paper which is worth investigating. Last but not least, as aforementioned, different from conventional secrecy communication, the channels of possible eavesdroppers, i.e., ERs, can be practically known at the transmitter in SWIPT systems, which motivates new transmission designs for physical layer security.

The rest of this paper is organized as follows. Section \ref{sec:system model} introduces the MISO SWIPT system model. Section \ref{sec:Problem Formulation} presents the formulations of the secrecy rate maximization problem subject to given harvested energy constraints as well as the weighted sum harvested energy maximization problem subject to a given secrecy rate constraint. Sections \ref{sec:Proposed Solutions to Problem (P1)} and \ref{sec:Proposed Solutions to Problem (P2)} present the optimal and suboptimal solutions for the two formulated problems, respectively. Section \ref{sec:Numerical Results} provides numerical results to compare the performances of various proposed schemes. Finally, Section \ref{sec:Conclusion} concludes the paper.

{\it Notation}: Scalars are denoted by lower-case letters, vectors
by bold-face lower-case letters, and matrices by
bold-face upper-case letters. $\mv{I}$ and $\mv{0}$  denote an
identity matrix and an all-zero matrix, respectively, with
appropriate dimensions. For a square matrix $\mv{S}$, ${\rm Tr}(\mv{S})$
denotes the trace of $\mv{S}$; $\mv{S}\succeq\mv{0}$ ($\mv{S}\preceq \mv{0}$) and $\mv{S}\succ \mv{0}$ ($\mv{S}\prec \mv{0}$) mean that $\mv{S}$ is positive (negative) semi-definite and positive (negative) definite, respectively. For a matrix
$\mv{M}$ of arbitrary size, $\mv{M}^{H}$ and ${\rm rank}(\mv{M})$ denote the
conjugate transpose and rank of $\mv{M}$, respectively. $E[\cdot]$ denotes the statistical expectation. The
distribution of a circularly symmetric complex Gaussian (CSCG) random vector with mean vector $\mv{x}$ and
covariance matrix $\mv{\Sigma}$ is denoted by
$\mathcal{CN}(\mv{x},\mv{\Sigma})$; and $\sim$ stands for
``distributed as''. $\mathbb{C}^{x \times y}$ denotes the space of
$x\times y$ complex matrices. $\|\mv{x}\|$ denotes the Euclidean norm of a complex vector
$\mv{x}$.

\section{System Model}\label{sec:system model}

We consider a multiuser MISO downlink system for SWIPT over a given frequency band as shown in Fig. \ref{fig2}. It is assumed that there is one single IR, and $K$ ERs denoted by the set $\mathcal{K}_{{\rm EH}}=\{{\rm ER}_1,\cdots,{\rm ER}_K\}$. The ERs are assumed to be closer to the transmitter (Tx) than the IR in order to harvest enough energy from the Tx. Suppose that Tx is equipped with $M>1$ antennas, while each IR/ER is equipped with one single antenna. We assume linear transmit beamforming at Tx, where the IR is assigned with one dedicated information beam, while the $K$ ERs are in total assigned with $d\leq M$ energy beams without loss of generality. Therefore, the complex baseband transmitted signal of Tx can be expressed as
\begin{align}\label{eqn:signal3}
\mv{x}=\mv{v}_0s_0+\sum\limits_{i=1}^d\mv{w}_i s_i,
\end{align}where $\mv{v}_0\in \mathbb{C}^{M\times 1}$ and $\mv{w}_i\in \mathbb{C}^{M\times 1}$ denote the information beamforming vector and the $i$th energy beamforming vector, $1\leq i \leq d$, respectively; $s_0$ denotes the transmitted signal for the IR, while $s_i$'s, $i=1,\cdots,d$, denote the energy-carrying signals for energy beams. It is assumed that $s_0$ is a CSCG random variable with zero mean and unit variance, denoted by $s_0\sim \mathcal{CN}(0,1)$. Furthermore, $s_i$'s, $1\leq i \leq d$, in general can be arbitrary independent random signals each with unit average power. Since in this paper we consider secret information transmission to the IR, the energy signals $s_i$, $1\leq i \leq d$, also play the role of AN to reduce the information rate eavesdropped by ERs \cite{Goel06}. As a result, similarly as \cite{Goel06}, we assume that $s_i$'s are independent and identically distributed (i.i.d.) CSCG random variables denoted by $s_i\sim \mathcal{CN}(0,1)$, $\forall i$, since the worst-case noise distribution for the eavesdropping ERs is known to be Gaussian. Suppose that Tx has a transmit sum-power constraint $\bar{P}$; from (\ref{eqn:signal3}), we thus have $E[\mv{x}^H\mv{x}]=\|\mv{v}_0\|^2+\sum_{i=1}^d\|\mv{w}_i\|^2\leq \bar{P}$.

In this paper, we assume a quasi-static fading environment and for convenience denote $\mv{h}\in \mathbb{C}^{M\times 1}$ and $\mv{g}_k\in \mathbb{C}^{M\times 1}$ as the conjugated complex channel vectors from Tx to IR and ${\rm ER}_k$, $k=1,\cdots,K$, respectively, where $\mv{h}$ and $\mv{g}_k$'s are assumed to be linearly independent. Note that in the case of $K>M$, linear independence in this paper implies that for any $M\times M$ matrix $\tilde{\mv{H}}$, in which the $M$ row vectors constitute any subset of channel vectors from $\mv{h}^H$ and $\mv{g}_k^H$'s, we have ${\rm rank}(\tilde{\mv{H}})=M$. Furthermore, let $\rho_h^2=\|\mv{h}\|^2/M$ and $\rho_{g_k}^2=\|\mv{g}_k\|^2/M$ denote the average per-antenna power of the IR's and ${\rm ER}_{k}$'s channels, respectively; then it is assumed that $\rho_{g_k}^2> \rho_h^2$, $\forall k$ , to be consistent with the receiver-location based transmission scheduling (cf. Fig. \ref{fig2}). It is further assumed that $\mv{h}$ and $\mv{g}_k$'s are perfectly known at Tx. The signal received at IR is then expressed as
\begin{align}\label{eqn:signal1}
\begin{small}y_0=\mv{h}^H\mv{x}+z_0,\end{small}
\end{align}where $z_0\sim \mathcal{CN}(0,\sigma_0^2)$ denotes the additive noise at IR. Furthermore, the signal received at ${\rm ER}_k$ can be expressed as
\begin{align}\label{eqn:signal2}
\begin{small}y_k=\mv{g}_k^H\mv{x}+z_k, ~~~ k=1,\cdots,K,\end{small}
\end{align}where $z_k\sim \mathcal{CN}(0,\sigma_k^2)$ denotes the additive noise at ${\rm ER}_k$. It is assumed that $z_k$'s are independent over $k$.

According to (\ref{eqn:signal1}), the SINR at IR can be expressed as
\begin{align}\label{eqn:IR SINR}
\begin{small}{\rm SINR}_0=\frac{|\mv{v}_0^H\mv{h}|^2}{\sum\limits_{i=1}^d|\mv{w}_i^H\mv{h}|^2+\sigma_0^2}. \end{small}
\end{align}

From (\ref{eqn:signal2}), the SINR at ${\rm ER}_k$ (suppose that it is an eavesdropper to decode the message for the IR instead of harvesting energy) can be expressed as
\begin{align}\label{eqn:ER SINR}
\begin{small}{\rm SINR}_k=\frac{|\mv{v}_0^H\mv{g}_k|^2}{\sum\limits_{i=1}^d|\mv{w}_i^H\mv{g}_k|^2+\sigma_k^2}, ~~~ k=1,\cdots,K. \end{small}
\end{align}The achievable secrecy rate at IR is thus given by \cite{Poor07}:
\begin{align}\label{eqn:secrecy rate}
\begin{small}r_0=\min\limits_{1\leq k\leq K} ~~~ \log_2\left(1+{\rm SINR}_0\right)- \log_2\left(1+{\rm SINR}_k \right). \end{small}
\end{align}Notice that the above achievable rate may be a conservative one in practical SWIPT systems since it is unlikely that all ERs will not harvest energy but instead eavesdrop information for the IR.

On the other hand, for wireless power transfer, due to the broadcast nature of wireless channels, the energy carried by all information and energy beams, i.e., $\mv{v}_0$ and $\mv{w}_i$'s ($1\leq i \leq d$), can all be harvested at each ER. Hence, assuming unit slot duration, the harvested energy of ${\rm ER}_k$ in each slot is given by \cite{Rui11}:
\begin{align}\label{eqn:harvested energy}
\begin{small}E_k=\zeta \left(|\mv{v}_0^H\mv{g}_k|^2+\sum\limits_{i=1}^d|\mv{w}_i^H\mv{g}_k|^2\right), ~~~ 1\leq k \leq K, \end{small}
\end{align}where $0<\zeta\leq 1$ denotes the energy harvesting efficiency.

\section{Problem Formulation}\label{sec:Problem Formulation}

In this paper, two secrecy beamforming design problems are considered as follows. First, we aim to maximize the secrecy rate of the IR subject to individual harvested energy constraints for all ERs. The first design problem is thus given by\begin{align*}\mathrm{(P1)}:~\mathop{\mathtt{Maximize}}_{\mv{v}_0,\{\mv{w}_i\}}
& ~~  r_0 \\
\mathtt {Subject \ to} & ~~ \zeta \left(|\mv{v}_0^H\mv{g}_k|^2+\sum\limits_{i=1}^d|\mv{w}_i^H\mv{g}_k|^2\right)\geq \bar{E}_k, ~ \forall k, \\ & ~~
\|\mv{v}_0\|^2+\sum\limits_{i=1}^d\|\mv{w}_i\|^2 \leq \bar{P},
\end{align*}where $\bar{E}_k$ denotes the harvested energy constraint for ${\rm ER}_k$. Note that (P1) is applicable for the scenario when ERs have strict energy harvesting requirements while the IR only requires an opportunistic information transmission.

Also we are interested in maximizing the weighted sum-energy transferred to ERs subject to a given secrecy rate constraint for IR. Therefore, the problem is formulated as
\begin{align*}\mathrm{(P2)}:~\mathop{\mathtt{Maximize}}_{\mv{v}_0,\{\mv{w}_i\}}
& ~~~  \sum\limits_{k=1}^K \mu_k\zeta \left(|\mv{v}_0^H\mv{g}_k|^2+\sum\limits_{i=1}^d|\mv{w}_i^H\mv{g}_k|^2\right)  \\
\mathtt {Subject \ to} & ~~~ r_0 \geq \bar{r}_0 , \\ & ~~~
\|\mv{v}_0\|^2+\sum\limits_{i=1}^d\|\mv{w}_i\|^2 \leq \bar{P},
\end{align*}where $\mu_k\geq 0$ denotes the energy weight for ${\rm ER}_k$, and $\bar{r}_0$ is the target secrecy rate for IR. (P2) applies for the scenario when the IR has a stringent rate requirement (e.g. delay-limited transmission) but ERs require opportunistic energy harvesting (with different priorities).

Notice that there are two conflicting goals in designing the information beamforming vector $\mv{v}_0$ for both problems (P1) and (P2). On one hand, to maximize the harvested energy at each ER, the power of the received signal at ${\rm ER}_k$ due to the information beam, i.e., $|\mv{v}_0^H\mv{g}_k|^2$, is desired to be as large as possible. However, on the other hand, from the viewpoint of secrecy rate maximization according to (\ref{eqn:IR SINR})--(\ref{eqn:secrecy rate}), it follows that $|\mv{v}_0^H\mv{g}_k|^2$ should be minimized at each ${\rm ER}_k$ to avoid any ``leakage'' information. To resolve the above conflict, we need to properly design the energy beamforming vectors $\mv{w}_i$, $i=1,\cdots,d$, since they not only provide direct wireless energy transfer to ERs, but also play the important role of AN to reduce the ERs' SINR in (\ref{eqn:ER SINR}) for decoding the IR's message.

Since both the secrecy rate $r_0$ for IR given in (\ref{eqn:secrecy rate}) and the harvested energy $E_k$ of ${\rm ER}_k$ given in (\ref{eqn:harvested energy}) are non-concave functions with respect to $\mv{v}_0$ and $\mv{w}_i$'s, problems (P1) and (P2) are both non-convex in general, and thus the strong duality does not apply for them \cite{Boyd04}. As a result, (P1) and (P2) do not have equivalent solutions under Lagrangian duality. In the following two sections, we address the solutions to these two problems, respectively.

\section{Proposed Solutions to Problem (P1)}\label{sec:Proposed Solutions to Problem (P1)}

In this section, we present both optimal and suboptimal solutions to (P1). First, we study the feasibility of this problem by setting $r_0=0$ to (P1) with $\mv{v}_0=\mv{0}$ for given $\bar{E}_k$'s and $\bar{P}$, i.e,\begin{align*}\mathrm{(P1-NoIT)}:\nonumber \\ \mathop{\mathtt{Maximize}}_{\{\mv{w}_i\}}
& ~~~  0 \\
\mathtt {Subject \ to} & ~~~ \zeta \left(\sum\limits_{i=1}^d|\mv{w}_i^H\mv{g}_k|^2\right)\geq \bar{E}_k, ~~~ 1\leq k \leq K, \\ & ~~~
\sum\limits_{i=1}^d\|\mv{w}_i\|^2 \leq \bar{P}.
\end{align*}Note that this problem corresponds to the case when there is no information transmission to IR, and thus $\mv{w}_i$'s play the only role of energy beams. Define $\mv{Q}=\sum_{i=1}^d\mv{w}_i\mv{w}_i^H$. Then we can reformulate (P1-NoIT) as a semidefinite programming (SDP) given by
\begin{align*}\mathrm{(P1-NoIT-SDP)}:\nonumber \\ \mathop{\mathtt{Maximize}}_{\mv{Q}}
& ~~~  0 \\
\mathtt {Subject \ to} & ~~~ \zeta {\rm Tr}(\mv{G}_k\mv{Q}) \geq \bar{E}_k, ~~~ 1\leq k \leq K, \\ & ~~~
{\rm Tr}(\mv{Q}) \leq \bar{P},¡¡\\ &¡¡~~~¡¡\mv{Q} \succeq \mv{0},
\end{align*}where $\mv{G}_k=\mv{g}_k\mv{g}_k^H$. Note that if there exists a feasible solution $\mv{Q}^\ast$ to (P1-NoIT-SDP), then with $d={\rm rank}(\mv{Q}^\ast)$, the energy beams $\mv{w}_i^\ast$, $i=1,\cdots,d$, obtained by the eigenvalue decomposition (EVD) of $\mv{Q}^\ast$, are also feasible to (P1-NoIT). Thereby, the feasibility of (P1) can be checked by solving (P1-NoIT-SDP) via existing software, e.g., CVX \cite{Boyd11}. Without loss of generality, in the rest of this paper, we assume that (P1) is feasible.

Next, we consider the other special case of (P1) when there is no energy transfer requirement, i.e., $\bar{E}_k=0$, $\forall k$. In this case, (P1) reduces to\begin{align*}\mathrm{(P1-NoET)}:~\mathop{\mathtt{Maximize}}_{\mv{v}_0,\{\mv{w}_i\}}
& ~~~  r_0 \\
\mathtt {Subject \ to} & ~~~ \|\mv{v}_0\|^2+\sum\limits_{i=1}^d\|\mv{w}_i\|^2 \leq \bar{P}.
\end{align*}Note that this problem is the conventional secrecy rate maximization under a MISO broadcast channel (BC) setup and has been solved in \cite{MaKent11}, for which the details are omitted for brevity. We will discuss the main difference of the optimal solution to (P1) from that of (P1-NoET) in \cite{MaKent11} due to the additional harvested energy constraints for ERs in Section \ref{sec:Optimal Solution to Problem (P1)} (see Remark \ref{remark3}).

\subsection{Optimal Solution to (P1)}\label{sec:Optimal Solution to Problem (P1)}


In this subsection, we propose a SDR-based algorithm to solve problem (P1) optimally by reformulating it into two sub-problems. First, similar to \cite{Rui10}, it can be shown that there always exists a SINR constraint $\gamma_e> 0$ at all ERs such that the following problem
\begin{align*}\mathrm{(P1.1)}:\nonumber \\ \mathop{\mathtt{Maximize}}_{\mv{v}_0,\{\mv{w}_i\}}
& ~ \frac{|\mv{v}_0^H\mv{h}|^2}{\sum\limits_{i=1}^d|\mv{w}_i^H\mv{h}|^2+\sigma_0^2} \\
\mathtt {Subject \ to} & ~ \frac{|\mv{v}_0^H\mv{g}_k|^2}{\sum\limits_{i=1}^d|\mv{w}_i^H\mv{g}_k|^2+\sigma_k^2}\leq \gamma_e, ~ 1\leq k \leq K,  \\ & ~ \zeta\left(|\mv{v}_0^H\mv{g}_k|^2+\sum\limits_{i=1}^d|\mv{w}_i^H\mv{g}_k|^2\right)\geq \bar{E}_k, ~ 1\leq k \leq K, \\ & ~
\|\mv{v}_0\|^2+\sum\limits_{i=1}^d\|\mv{w}_i\|^2 \leq \bar{P},
\end{align*}has the same optimal solution to (P1).

%
%
%

Let $g_1(\gamma_e)$ denote the optimal value of problem (P1.1) with a given $\gamma_e>0$. Then, it can be also shown that the optimal value of (P1) is the same as that of the following problem.
\begin{align*}\mathrm{(P1.2)}:~\mathop{\mathtt{Maximize}}_{\gamma_e>0}
& ~~~ \log_2 \left(\frac{1+g_1(\gamma_e)}{1+\gamma_e}\right).
\end{align*}


Let $\gamma_e^\ast$ denote the optimal solution to problem (P1.2). From the above results, with $\gamma_e=\gamma_e^\ast$, it follows that (P1) and (P1.1) have the same optimal solution. Therefore, (P1) can be solved in the following two steps: The optimal $\gamma_e$ for (P1.2) is obtained by one dimension search over $\gamma_e>0$, while given any $\gamma_e$, $g_1(\gamma_e)$ is obtained by solving (P1.1). Hence, in the rest of this subsection, we focus on solving (P1.1).

Note that (P1.1) is still non-convex. Define $\mv{S}=\mv{v}_0\mv{v}_0^H$ and $\mv{Q}=\sum_{i=1}^d\mv{w}_i\mv{w}_i^H$. Then it follows that ${\rm rank}(\mv{S})\leq 1$ and ${\rm rank}(\mv{Q})\leq d$. By ignoring the rank-one constraint on $\mv{S}$, the SDR of problem (P1.1) can be expressed as
\begin{align}\mathrm{(P1.1-SDR)}:\nonumber \\ \mathop{\mathtt{Maximize}}_{\mv{S},\mv{Q}}
& ~ \frac{{\rm Tr}(\mv{H}\mv{S})}{{\rm Tr} (\mv{H}\mv{Q})+\sigma_0^2} \nonumber \\
\mathtt {Subject \ to} & ~ {\rm Tr}(\mv{G}_k\mv{S})\leq \gamma_e({\rm Tr}(\mv{G}_k\mv{Q})+\sigma_k^2), ~ \forall k, \label{1} \\ & ~ \zeta({\rm Tr}(\mv{G}_k\mv{S})+{\rm Tr}(\mv{G}_k\mv{Q})) \geq \bar{E}_k, ~ \forall k, \label{2} \\ & ~
{\rm Tr}(\mv{S})+{\rm Tr}(\mv{Q}) \leq \bar{P}, \label{3} \\ & ~ \mv{S}\succeq \mv{0},  ~ \mv{Q}\succeq \mv{0}, \label{4}
\end{align}where $\mv{H}=\mv{h}\mv{h}^H$ and $\mv{G}_k=\mv{g}_k\mv{g}_k^H$. If the optimal solution to problem (P1.1-SDR), denoted by $\mv{S}^\ast$ and $\mv{Q}^\ast$, satisfies ${\rm rank}(\mv{S}^\ast)= 1$, then the optimal information beam $\mv{v}_0^\ast$ and the optimal $d={\rm rank}(\mv{Q}^\ast)$ number of energy beams $\mv{w}_i^\ast$, $i=1,\cdots,d$, for problem (P1.1) can be obtained from the EVDs of $\mv{S}^\ast$ and $\mv{Q}^\ast$, respectively; otherwise, if ${\rm rank}(\mv{S}^\ast)>1$, then the optimal value of problem (P1.1-SDR) only serves as an upper bound on that of problem (P1.1). In the following, we check whether ${\rm rank}(\mv{S}^\ast)=1$ always holds or not for (P1.1-SDR).

(P1.1-SDR) is non-convex since its objective function is non-concave over $\mv{S}$ and $\mv{Q}$. However, we can apply the Charnes-Cooper transformation \cite{Cooper62,Ma10} to reformulate (P1.1-SDR) as an equivalent convex problem.

\begin{lemma}\label{lemma6}
Problem (P1.1-SDR) is equivalent to the following problem.
\begin{align*}  \mathrm{(P1.1-SDR-Eqv)}: \ \ \ \ \ \ \ \ \ \ \ \ \ \ \ \ \ \ \ \ \ \ \ \ \ \ \ \ \ \ \ \ \ \ \ \ \ \end{align*}\vspace{-28pt}  \begin{align} \mathop{\mathtt{Maximize}}_{\mv{S},\mv{Q},t}
& ~ {\rm Tr}(\mv{H}\mv{S}) \nonumber \\
\mathtt {Subject \ to} & ~ {\rm Tr}(\mv{H}\mv{Q})+t\sigma_0^2=1, \label{eqn:new constraint1}\\ & ~ {\rm Tr}(\mv{G}_k\mv{S})\leq \gamma_e({\rm Tr}(\mv{G}_k\mv{Q})+t\sigma_k^2), ~ \forall k, \label{eqn:new constraint2}\\ & ~ \zeta({\rm Tr}(\mv{G}_k\mv{S})+{\rm Tr}(\mv{G}_k\mv{Q})) \geq t\bar{E}_k, ~ \forall k,\label{eqn:new constraint3}\\ & ~
{\rm Tr}(\mv{S})+{\rm Tr}(\mv{Q}) \leq t\bar{P}, \hspace{-10pt} \label{eqn:new constraint4}\\ & ~ \mv{S}\succeq \mv{0},  ~ \mv{Q}\succeq \mv{0}, ~ t>0. \label{eqn:new constraint5}
\end{align}
\end{lemma}

\begin{proof}
First, given any feasible solution $(\mv{S},\mv{Q})$ to problem (P1.1-SDR), it can be shown that with the solution $(\mv{S}/({\rm Tr}(\mv{H}\mv{Q})+\sigma_0^2),\mv{Q}/({\rm Tr}(\mv{H}\mv{Q})+\sigma_0^2),1/({\rm Tr}(\mv{H}\mv{Q})+\sigma_0^2))$, (P1.1-SDR-Eqv) achieves the same objective value as that of (P1.1-SDR). Second, given any feasible solution $(\mv{S},\mv{Q},t)$ to (P1.1-SDR-Eqv), it can be similarly shown that with the solution $(\mv{S}/t,\mv{Q}/t)$, (P1.1-SDR) achieves the same objective value as that of (P1.1-SDR-Eqv). Therefore, (P1.1-SDR) and (P1.1-SDR-Eqv) have the same optimal value. Lemma \ref{lemma6} is thus proved.
\end{proof}

According to Lemma \ref{lemma6}, we can obtain the optimal solution to (P1.1-SDR) by solving (P1.1-SDR-Eqv).

\begin{remark}
It is worth noting that in the literature, problem (P1.1-SDR-Eqv) belongs to the so-called ``separable SDP'' \cite{Palomar10} since there are more than one design variables, i.e., $\mv{S}$ and $\mv{Q}$. According to \cite[Theorem 3.2]{Palomar10}, there exists an optimal solution $(\mv{S}^\ast,\mv{Q}^\ast)$ to (P1.1-SDR-Eqv) such that ${\rm rank}^2(\mv{S}^\ast)+{\rm rank}^2(\mv{Q}^\ast) \leq 2K+2$, where $2K+2$ denotes the number of linear constraints given in (\ref{eqn:new constraint1})--(\ref{eqn:new constraint4}). However, this is not sufficient to show ${\rm rank}(\mv{S}^\ast)=1$ because $K$ in our problem can be arbitrarily large. As a result, the well-known result in \cite{Palomar10} for separable SDPs cannot be applied in our case to show the tightness of SDR in (P1.1-SDR-Eqv). It is also worth noting that for the special case of one design variable, the tightness condition of SDR has been widely studied in the literature \cite{Beck06}--\cite{Zhangshuzhong09}, and the best result known so far is for the tightness of SDR with up to four linear constraints \cite{Zhangshuzhong09}. However, since in (P1.1-SDR-Eqv) there are two variables and $2K+2>4$ constraints since $K>1$ in general, the results in \cite{Beck06}--\cite{Zhangshuzhong09} also cannot be applied to our problem. For more information about the tightness condition of SDR, the readers can refer to \cite{Luo10}.
\end{remark}

Since (P1.1-SDR-Eqv) is convex and satisfies the Slater's condition, its duality gap is zero \cite{Boyd04}. Let $\lambda$, $\{\beta_k\}$, $\{\alpha_k\}$, and $\theta$ denote the dual variables of (P1.1-SDR-Eqv) associated with the equality constraint in (\ref{eqn:new constraint1}), the SINR constraints of ERs in (\ref{eqn:new constraint2}), the harvested energy constraints of ERs in (\ref{eqn:new constraint3}), and the sum-power constraint in (\ref{eqn:new constraint4}), respectively. Then the Lagrangian of problem (P1.1-SDR-Eqv) is expressed as
\begin{align}
\begin{small}L_1(\mv{S},\mv{Q},\lambda,\{\beta_k\},\{\alpha_k\},\theta) = {\rm Tr}(\mv{A}_1\mv{S})+{\rm Tr}(\mv{B}_1\mv{Q}) +\xi_1t+\lambda, \end{small}\label{eqn:Lagrangian2}
\end{align}where\begin{align}
& \begin{small} \mv{A}_1=\mv{H}-\sum\limits_{k=1}^K\beta_k \mv{G}_k+\sum\limits_{k=1}^K\alpha_k\zeta\mv{G}_k-\theta\mv{I}, \end{small} \label{eqn:A2} \\
& \begin{small} \mv{B}_1=-\lambda\mv{H}+\sum\limits_{k=1}^K\beta_k\gamma_e\mv{G}_k+\sum\limits_{k=1}^K\alpha_k\zeta\mv{G}_k-\theta \mv{I}, \end{small} \label{eqn:B2} \\
& \begin{small} \xi_1=-\lambda\sigma_0^2+\sum\limits_{k=1}^K\beta_k\gamma_e \sigma_k^2-\sum\limits_{k=1}^K\alpha_k\bar{E}_k+\theta \bar{P}. \end{small} \label{eqn:pi2}
\end{align}

Let $\lambda^\ast$, $\{\beta_k^\ast\geq 0\}$, $\{\alpha_k^\ast \geq 0\}$, and $\theta^\ast\geq 0$ denote the optimal dual solutions to problem (P1.1-SDR-Eqv). Then, we have the following lemma.

\begin{lemma}\label{lemma7}
The optimal dual solution to problem (P1.1-SDR-Eqv) satisfies that $\lambda^\ast>0$ and $\theta^\ast>0$ when $\gamma_e>0$.
\end{lemma}

\begin{proof}
Please refer to Appendix \ref{appendix6}.
\end{proof}

With $\theta^\ast>0$, it follows that in the optimal solution of problem (P1.1-SDR-Eqv), the sum-power constraint (\ref{eqn:new constraint4}) must be satisfied with equality due to the complementary slackness \cite{Boyd04}. Define \begin{align}\mv{D}_1^\ast=-\lambda^\ast\mv{H}-\sum_{k=1}^K\beta_k^\ast\mv{G}_k+\sum_{k=1}^K\alpha_k^\ast\zeta \mv{G}_k-\theta^\ast\mv{I},\end{align}and $l_1={\rm rank}(\mv{D}_1^\ast)$. Furthermore, let $\mv{\Pi}_1\in \mathbb{C}^{M\times (M-l_1)}$ denote the orthogonal basis of the null space of $\mv{D}_1^\ast$, where $\mv{\Pi}_1=\mv{0}$ if $l_1=M$, and $\mv{\pi}_{1,n}$ denote the $n$th column of $\mv{\Pi}_1$. Then based on Lemma \ref{lemma7}, we have the following proposition.

\begin{proposition}\label{proposition4}
The optimal solution $(\mv{S}^\ast,\mv{Q}^\ast,t^\ast)$ to problem (P1.1-SDR-Eqv) satisfies the following conditions:
\begin{itemize}
\item[1.] ${\rm rank}(\mv{Q}^\ast)\leq \min(K,M)$;
\item[2.] $\mv{S}^\ast$ can be expressed as\begin{align}\label{eqn:feasible rank}\begin{small}\mv{S}^\ast=\sum\limits_{n=1}^{M-l_1}a_n\mv{\pi}_{1,n}\mv{\pi}_{1,n}^H+b\mv{\tau}_1\mv{\tau}_1^H,\end{small}\end{align}where $a_n\geq 0$, $\forall n$, $b>0$, and $\mv{\tau}_1\in \mathbb{C}^{M\times 1}$ has unit-norm and satisfies $\mv{\tau}_1^H\mv{\Pi}_1=\mv{0}$.
\item[3.] If $\mv{S}^\ast$ given in (\ref{eqn:feasible rank}) has the rank larger than one, i.e., there exists at least an $n$ such that $a_n>0$, then the following solution
\begin{align}
& \begin{small}\bar{\mv{S}}^\ast=b\mv{\tau}_1\mv{\tau}_1^H, \end{small} \label{eqn:new S}\\
& \begin{small}\bar{\mv{Q}}^\ast=\mv{Q}^\ast+\sum\limits_{n=1}^{M-l_1}a_n\mv{\pi}_{1,n}\mv{\pi}_{1,n}^H, \end{small}\label{eqn:new Q}\\
& \begin{small}\bar{t}^\ast=t^\ast, \end{small}\label{eqn:new t}
\end{align}with ${\rm rank}(\bar{\mv{S}}^\ast)=1$ is also optimal to problem (P1.1-SDR-Eqv).
\end{itemize}
\end{proposition}

\begin{proof}
Please refer to Appendix \ref{appendix7}.
\end{proof}

With Proposition \ref{proposition4}, we are ready to find the optimal solution to problem (P1.1-SDR) with a rank-one covariance matrix for $\mv{S}$ as follows. First, we solve (P1.1-SDR-Eqv) via CVX. If the obtained solution $(\mv{S}^\ast,\mv{Q}^\ast,t^\ast)$ satisfies that ${\rm rank}(\mv{S}^\ast)=1$, then $(\mv{S}^\ast/t^\ast,\mv{Q}^\ast/t^\ast)$ will be the optimal solution to (P1.1-SDR) according to Lemma \ref{lemma6}. Otherwise, if ${\rm rank}(\mv{S}^\ast)>1$, we can construct a new solution $(\bar{\mv{S}}^\ast,\bar{\mv{Q}}^\ast,\bar{t}^\ast)$ with ${\rm rank}(\bar{\mv{S}}^\ast)=1$ according to (\ref{eqn:feasible rank})--(\ref{eqn:new t}). Then, $(\bar{\mv{S}}^\ast/\bar{t}^\ast,\bar{\mv{Q}}^\ast/\bar{t}^\ast)$ will be the optimal solution to (P1.1-SDR). Therefore, the rank-one relaxation on $\mv{S}$ in (P1.1-SDR) results in no loss of optimality to (P1.1), and given any $\gamma_e>0$, the value of $g_1(\gamma_e)$ can be obtained by solving (P1.1-SDR-Eqv). Furthermore, since ${\rm rank}(\mv{Q}^\ast)\leq \min(K,M)$ in Proposition \ref{proposition4}, it implies that in the case of $K<M$, at most $K$ energy beams are needed in the optimal solution of (P1), i.e., $d\leq K$.

It is worth noting that in general, Proposition \ref{proposition4} only guarantees the existence of a rank-one optimal covariance solution $\mv{S}^\ast$ to (P1.1-SDR-Eqv) and thus (P1.1). One interesting question is thus under what conditions the rank-one solution $\mv{S}^\ast$ to (P1.1-SDR-Eqv) is unique. To answer this question, we define the following two sets as
\begin{align}
\begin{small}\Psi=\{k|\beta_k^\ast=0, ~ k=1,\cdots,K\}, \end{small}\label{psi1} \\
\begin{small}\bar{\Psi}=\{k|\beta_k^\ast>0, ~ k=1,\cdots,K\}. \end{small}\label{psi2}
\end{align}Moreover, define $\bar{\mv{G}}_1=[\mv{h},\{\mv{g}_k\}]^H$, $\forall k\in \bar{\Psi}$, in which the row vectors consist of the channels of IR, i.e., $\mv{h}^H$, and a subset of ERs, i.e., $\mv{g}_k^H$'s, whose SINR constraints are tight in the optimal solution to (P1.1-SDR-Evq) with $\beta_k^\ast>0$ due to the complimentary slackness. Then we have the following proposition.

\begin{proposition}\label{proposition3}
The optimal solution $(\mv{S}^\ast,\mv{Q}^\ast,t^\ast)$ to (P1.1-SDR-Eqv) always satisfies that ${\rm rank}(\mv{S}^\ast)=1$ if there is no non-trivial (non-zero) solution $\mv{x}\in \mathbb{C}^{M\times 1}$ to the following equations:
\begin{align}\label{eqn:sufficient condition}
\begin{small}\left\{\begin{array}{l}\mv{x}^H\left(\sum\limits_{k\in \Psi}\alpha_k^\ast\zeta\mv{G}_k-\theta^\ast\mv{I}\right)\mv{x}=0, \\ \bar{\mv{G}}_1\mv{x}=\mv{0}. \end{array} \right.\end{small}
\end{align}
\end{proposition}

\begin{proof}
Please refer to Appendix \ref{appendix5}.
\end{proof}

Let $|\bar{\Psi}|$ denote the cardinality of the set $\bar{\Psi}$. Based on Proposition \ref{proposition3}, we then have the following two corollaries.

\begin{corollary}\label{corollary1}
In the case of $K<M-1$, if $|\bar{\Psi}|=K$, i.e., the SINR constraint (\ref{eqn:new constraint2}) is tight for all ERs, then the optimal solution $(\mv{S}^\ast,\mv{Q}^\ast,t^\ast)$ to (P1.1-SDR-Eqv) always satisfies that ${\rm rank}(\mv{S}^\ast)=1$.
\end{corollary}

\begin{proof}
If $|\bar{\Psi}|=K$, i.e., $\Psi=\emptyset$, then we have $\sum_{k\in \Psi}\alpha_k^\ast \zeta\mv{G}_k-\theta^\ast\mv{I}=-\theta^\ast\mv{I}$. Since $\theta^\ast>0$ according to Lemma \ref{lemma7}, there is no non-zero solution to the equation $-\theta^\ast\mv{x}^H\mv{x}=0$. According to Proposition \ref{proposition3}, Corollary \ref{corollary1} is thus proved.
\end{proof}

\begin{corollary}\label{corollary2}
In the case of $K\geq M-1$, if $|\bar{\Psi}|\geq M-1$, then the optimal solution $(\mv{S}^\ast,\mv{Q}^\ast,t^\ast)$ to (P1.1-SDR-Eqv) always satisfies that ${\rm rank}(\mv{S}^\ast)=1$.
\end{corollary}

\begin{proof}
Since all the channels are assumed to be linearly independent, if $|\bar{\Psi}|\geq M-1$, then we have ${\rm rank}(\bar{\mv{G}}_1)=M$, and thus there is no non-zero solution to the equation $\bar{\mv{G}}_1\mv{x}=\mv{0}$. According to Proposition \ref{proposition3}, Corollary \ref{corollary2} is thus proved.
\end{proof}


According to the complementary slackness, if $\beta_k^\ast>0$, then the SINR constraint for ${{\rm ER}}_k$ must be tight in problem (P1.1-SDR-Eqv). In other words, $|\bar{\Psi}|$ denotes the number of ERs whose SINR constraints are active in the optimal solution to (P1.1-SDR-Eqv). Therefore, Corollaries \ref{corollary1} and \ref{corollary2} imply that if $|\bar{\Psi}|\geq \min(M-1,K)$, i.e., the SINR constraint in (\ref{eqn:new constraint2}) is tight for at least $\min(M-1,K)$ ERs, ${\rm rank}(\mv{S}^\ast)=1$ must hold for (P1.1-SDR-Eqv). Note that in our assumed system setup, all ERs are closer to the Tx than the IR and as a result, they all have better channels for eavesdropping the IR's message. It is thus expected that the SINR constraint in (\ref{eqn:new constraint2}) should be active for all ERs with a very high probability. Therefore, the condition given in Corollaries \ref{corollary1} and \ref{corollary2}, i.e., $|\bar{\Psi}|\geq \min(M-1,K)$, can be considered to be practically satisfied, under which the uniqueness of the rank-one optimal covariance solution $\mv{S}^\ast$ to (P1.1-SDR-Eqv) also holds.

\begin{remark}\label{remark3}
It is worth pointing out two main differences in the optimal beamforming solution to (P1) with versus without the energy harvesting (EH) constraints. First, consider the case of one single ER, i.e., $K=1$, in (P1). In this case, without the EH constraint, it has been shown in \cite{Liu10}, \cite{Oggier11} that the secrecy capacity for the IR is given by $C_{\rm s}=\max\limits_{\|\mv{v}_0\|^2\leq \bar{P}} ~~ \log_2\left(1+\frac{|\mv{v}_0^H\mv{h}|^2}{\sigma_0^2}\right)-\log_2\left(1+\frac{|\mv{v}_0^H\mv{g}_1|^2}{\sigma_1^2}\right)$. Notice that AN is not needed in achieving $C_{\rm s}$, which is also the optimal value of (P1-NoET) in the case of $K=1$, and the optimal beamforming solution for the IR to achieve $C_{\rm s}$ has been obtained in \cite{Rui10}, \cite{Ulukus07}. In contrast, with the EH constraint added to (P1), in order to deliver the required wireless energy to the ER and at the same time achieve the maximum secrecy rate for the IR, AN is in general needed according to Proposition \ref{proposition4}, since it follows that ${\rm rank}(\mv{Q}^\ast)\leq 1$ in the case of $K=1$, i.e., one energy beam is in general needed to power the ER and in the meanwhile carry the AN to interfere with it from eavesdropping the IR's message. Second, consider the more general case with multiple ERs, i.e., $K>1$. In this case, without considering EH constraints at ERs, (P1) reduces to (P1-NoET), which has been solved in \cite{MaKent11}. It was shown in \cite{MaKent11}, \cite{Ma11} that if all the channels are linearly independent, SDR can always obtain the unique optimal rank-one covariance (beamforming) solution for the IR. However, with the additional EH constraints added for ERs, it is in general not always true that SDR yields a rank-one covariance solution for IR, as shown in Proposition \ref{proposition4}. However, we are able to show in Proposition \ref{proposition4} that the optimal rank-one covariance solution for IR always exists and can be obtained by a simple reconstruction of the optimal solution.

\end{remark}

\begin{remark}\label{remark5}
It is also worth noting that the optimal solution obtained for (P1) is applicable for the special case of maximizing the IR's rate but without considering the secret transmission as given by the following problem.
\begin{align*}\mathrm{({\rm P1-NoSC})}:\nonumber \\ \mathop{\mathtt{Maximize}}_{\mv{v}_0,\{\mv{w}_i\}}
& ~  \begin{small}\log_2\left(1+\frac{|\mv{v}_0^H\mv{h}|^2}{\sum\limits_{i=1}^d|\mv{w}_i^H\mv{h}|^2+\sigma_0^2}\right) \end{small} \\
\mathtt {Subject \ to} & ~ \zeta \left(|\mv{v}_0^H\mv{g}_k|^2+\sum\limits_{i=1}^d|\mv{w}_i^H\mv{g}_k|^2\right) \geq \bar{E}_k, ~ \forall k, \\ & ~
\|\mv{v}_0\|^2+\sum\limits_{i=1}^d\|\mv{w}_i\|^2 \leq \bar{P}.
\end{align*}It is observed that (P1-NoSC) is equivalent to (P1.1) by setting $\gamma_e=\infty$ in the first constraint. Therefore, similar to problem (P1.1), SDR can be applied to obtain the optimal beamforming solution to (P1-NoSC).
\end{remark}

\subsection{Suboptimal Solutions to (P1)}\label{sec:Suboptimal Solutions to Problem (P1)}

The optimal solution to (P1) proposed in Section \ref{sec:Optimal Solution to Problem (P1)} requires a joint optimization of the information/energy beamforming vectors and their power allocation. In this subsection, we propose two suboptimal solutions for (P1) which can be designed with lower complexity than the optimal solution. Similar to \cite{Goel06}, in our proposed suboptimal solutions, the energy beams $\mv{w}_i$, $i=1,\cdots,d$, are all restricted to lie in the null space of the IR's channel $\mv{h}$ such that they cause no interference to IR. However, the information beam $\mv{v}_0$ is aligned to the null space of the ERs' channels $\mv{G}=[\mv{g}_1,\cdots,\mv{g}_K]^H$ in the first suboptimal solution in order to eliminate the information leaked to all ERs, but to the same direction as $\mv{h}$ in the second suboptimal solution to maximize the IR's SINR. Note that the first suboptimal solution is only applicable when $K<M$ since otherwise the null space of $\mv{G}$ is empty. In the following, we present the two proposed suboptimal solutions in more details.

\subsubsection{Suboptimal Solution \uppercase\expandafter{\romannumeral1}}

\ \ \ \ \

Supposing that $K<M$, then the first suboptimal solution aims to solve problem (P1) with the additional constraints: $\mv{v}_0^H\mv{g}_k=0$, $\forall k$, and $\mv{w}_i^H\mv{h}=0$, $\forall i$. Consider first the information beam $\mv{v}_0$. Let the singular value decomposition (SVD) of $\mv{G}$ be denoted as \vspace{-3pt}
\begin{align}
\begin{small}\mv{G}=\mv{U}\mv{\Lambda}\mv{V}^H=\mv{U}\mv{\Lambda}[\bar{\mv{V}} \ \tilde{\mv{V}}]^H, \end{small}
\end{align}where $\mv{U}\in \mathbb{C}^{K\times K}$ and $\mv{V}\in \mathbb{C}^{M\times M}$ are unitary matrices, i.e., $\mv{U}\mv{U}^H=\mv{U}^H\mv{U}=\mv{I}$, $\mv{V}\mv{V}^H=\mv{V}^H\mv{V}=\mv{I}$, and $\mv{\Lambda}$ is a $K\times M$ rectangular diagonal matrix. Furthermore, $\bar{\mv{V}}\in \mathbb{C}^{M\times K}$ and $\tilde{\mv{V}}\in \mathbb{C}^{M\times (M-K)}$ consist of the first $K$ and the last $M-K$ right singular vectors of $\mv{G}$, respectively. It can be shown that $\tilde{\mv{V}}$ with $\tilde{\mv{V}}^H\tilde{\mv{V}}=\mv{I}$ forms an orthogonal basis for the null space of $\mv{G}$. Thus, to guarantee that $\mv{G}\mv{v}_0=\mv{0}$, $\mv{v}_0$ must be in the following form:
\begin{align}\label{eqn:sub v0}
\begin{small}\mv{v}_0=\sqrt{\tilde{P}_0}\tilde{\mv{V}}\tilde{\mv{v}}_0, \end{small}
\end{align}where $\tilde{P}_0=\|\mv{v}_0\|^2$ denotes the transmit power of the information beam, and $\tilde{\mv{v}}_0$ is an arbitrary $(M-K)\times 1$ complex vector of unit norm. It can be shown that to maximize the IR's received power, $\tilde{\mv{v}}_0$ should be aligned to the same direction as the equivalent channel $\tilde{\mv{V}}^H\mv{h}$, i.e., $\tilde{\mv{v}}_0^\ast=\tilde{\mv{V}}^H\mv{h}/\|\tilde{\mv{V}}^H\mv{h}\|$. Given that all the energy beams are aligned to the null space of $\mv{h}$, the achievable secrecy rate of IR is expressed as
\begin{align}\label{eqn:sub rate1}
\begin{small}r_0^{({\rm I})}=\log_2\left(1+\frac{\tilde{P}_0\|\tilde{\mv{V}}^H\mv{h}\|^2}{\sigma_0^2}\right). \end{small}
\end{align}

Next, consider the energy beam $\mv{w}_i$'s. Define the projection matrix as $\mv{T}=\mv{I}-\mv{h}\mv{h}^H/\|\mv{h}\|^2$. Without loss of generality, we can express $\mv{T}=\tilde{\mv{X}}\tilde{\mv{X}}^H$, where $\tilde{\mv{X}}\in \mathbb{C}^{M\times (M-1)}$ satisfies $\tilde{\mv{X}}^H\tilde{\mv{X}}=\mv{I}$. It can be shown that $\tilde{\mv{X}}$ forms an orthogonal basis for the null space of $\mv{h}^H$. Thus, to guarantee that $\mv{h}^H\mv{w}_i=0$, $\forall i$, $\mv{w}_i$ must be in the following form:
\begin{align}\label{eqn:sub wi}
\mv{w}_i=\tilde{\mv{X}}\tilde{\mv{w}}_i, ~~~ i=1,\cdots,d,
\end{align}where $\tilde{\mv{w}}_i$ is an arbitrary $(M-1)\times 1$ complex vector. In this case, the energy harvested at ${\rm ER}_k$ is thus given by\begin{align}\label{eqn:sub1 energy}E_k^{({\rm I})}=\zeta\sum_{i=1}^d|\mv{w}_i^H\mv{g}_k|^2=\zeta\left(\sum_{i=1}^d\tilde{\mv{w}}_i^H\tilde{\mv{G}}_k\tilde{\mv{w}}_i\right), ~~~ 1\leq k \leq K,\end{align}where $\tilde{\mv{G}}_k=\tilde{\mv{X}}^H\mv{G}_k\tilde{\mv{X}}$.

It can be observed from (\ref{eqn:sub rate1}) that to maximize the secrecy rate $r_0^{({\rm I})}$, $\tilde{P}_0$ should be as large as possible. Therefore, the optimal energy beams can be obtained by solving the following problem.
\begin{align*}\mathrm{(P1-Sub1)}: \\ \mathop{\mathtt{Minimize}}_{\{\tilde{\mv{w}}_i\}}
& ~~~ \sum\limits_{i=1}^d\|\tilde{\mv{w}}_i\|^2 \\
\mathtt {Subject \ to} & ~~~ \zeta\left(\sum\limits_{i=1}^d\tilde{\mv{w}}_i^H\tilde{\mv{G}}_k\tilde{\mv{w}}_i\right)\geq \bar{E}_k, ~~~ 1\leq k \leq K.
\end{align*}Define $\tilde{\mv{Q}}=\sum_{i=1}^d\tilde{\mv{w}}_i\tilde{\mv{w}}_i^H$. Then the SDP reformulation of (P1-Sub1) can be expressed as
\begin{align*}\mathrm{(P1-Sub1-SDP)}:\\ \mathop{\mathtt{Minimize}}_{\tilde{\mv{Q}}}
& ~~~ {\rm Tr}(\tilde{\mv{Q}}) \\
\mathtt {Subject \ to} & ~~~ \zeta{\rm Tr}(\tilde{\mv{G}}_k \tilde{\mv{Q}})\geq \bar{E}_k, ~~~ 1\leq k \leq K.
\end{align*}Assuming that (P1-Sub1-SDP) is feasible, then it can be solved via CVX. Denote the optimal solution for this problem as $\tilde{\mv{Q}}^\ast$. Then the optimal solution to (P1-Sub1), denoted by $\tilde{\mv{w}}_i^\ast$'s, can be obtained by the EVD of $\tilde{\mv{Q}}^\ast$, and the corresponding optimal energy beams can be obtained as \begin{align}\label{eqn:energy beam 2}\mv{w}_i^\ast=\tilde{\mv{X}}\tilde{\mv{w}}_i^\ast, ~~~ 1\leq i \leq d,\end{align}where $d={\rm rank}(\tilde{\mv{Q}}^\ast)$. Furthermore, the optimal power allocation for the information beam is given by  $\tilde{P}_0^\ast=\bar{P}-\sum_{i=1}^d\|\mv{w}_i^\ast\|^2$; by assuming $\tilde{P}^\ast> 0$, the optimal information beam is thus obtained as\begin{align}\label{eqn:infor beam2}\mv{v}_0^\ast=\sqrt{\tilde{P}_0^\ast}\tilde{\mv{V}}\tilde{\mv{v}}_0^\ast=\frac{\sqrt{\bar{P}-\sum\limits_{i=1}^d\|\mv{w}_i^\ast\|^2}\tilde{\mv{V}}\tilde{\mv{V}}^H\mv{h}}{\|\tilde{\mv{V}}^H\mv{h}\|}.\end{align}

\subsubsection{Suboptimal Solution \uppercase\expandafter{\romannumeral2}}

\ \ \ \ \

The second suboptimal solution aims to solve problem (P1) with the additional constraints: $\mv{v}_0=\sqrt{\hat{P}_0}\mv{h}/\|\mv{h}\|$ and $\mv{w}_i^H\mv{h}=0$, $\forall i$, where $\hat{P}_0=\|\mv{v}_0\|^2$ denotes the transmit power allocated to the information beam. To reduce the design complexity in this case, we further assume that the energy beams are in the form of $\mv{w}_i=\sqrt{\bar{P}-\hat{P}_0}\mv{w}_i^\ast/\sqrt{\sum_{i=1}^d\|\mv{w}_i^\ast\|^2}$, $1\leq i \leq d$, where $\mv{w}_i^\ast$'s are the energy beams obtained by Suboptimal Solution \uppercase\expandafter{\romannumeral1} in (\ref{eqn:energy beam 2}). Therefore, the secrecy rate of the IR is expressed as
\begin{align}
& \begin{small} r_0^{({\rm II})}=\min\limits_{1\leq k \leq K} ~ \log_2\left(1+\frac{\hat{P}_0\|\mv{h}\|^2}{\sigma_0^2}\right)-\nonumber \end{small} \\ &\begin{small}\log_2\left(1+\frac{\hat{P}_0|\mv{h}^H\mv{g}_k|^2}{\|\mv{h}\|^2\left(\sum\limits_{i=1}^d((\bar{P}-\hat{P}_0)|\mv{g}_k^H\mv{w}_i^\ast|^2/\sum\limits_{j=1}^d\|\mv{w}_j^\ast\|^2)+\sigma_k^2\right)}\right).\end{small}
\end{align}Furthermore, the harvested energy by ${\rm ER}_k$ is expressed as
\begin{align}
\begin{small}E_k^{({\rm II})}=\zeta\left(\frac{\hat{P}_0|\mv{h}^H\mv{g}_k|^2}{\|\mv{h}\|^2}+\sum\limits_{i=1}^d\frac{(\bar{P}_0-\hat{P}_0)|\mv{g}_k^H\mv{w}_i^\ast|^2}{\sum\limits_{j=1}^d\|\mv{w}_j^\ast\|^2}\right), ~ \forall k.\end{small}
\end{align}Define the set of feasible power allocation for the information beam as $\hat{\mathcal{P}}_0=\{\hat{P}_0|E_k^{({\rm II})}\geq \bar{E}_k, \ 1\leq k \leq K, \ 0<\hat{P}_0\leq \bar{P}\}$, which is assumed to be non-empty. To maximize the secrecy rate of the IR subject to individual harvested energy constraints of ERs, we need to solve the following problem. \begin{align*}\mathrm{(P1-Sub2)}:~\mathop{\mathtt{Maximize}}_{\hat{P}_0}
& ~~~ r_0^{({\rm II})} \\
\mathtt {Subject \ to} & ~~~ \hat{P}_0\in \hat{\mathcal{P}}_0.
\end{align*}The optimal solution to (P1-Sub2), denoted by $\hat{P}_0^\ast$, can be obtained by a one-dimension search over the set $\hat{\mathcal{P}}_0$.

\section{Proposed Solutions to Problem (P2)}\label{sec:Proposed Solutions to Problem (P2)}\vspace{-2pt}
In this section, we present the optimal solution as well as two suboptimal solutions to (P2). Similar to (P1), we first study the feasibility of (P2) for a given pair of $\bar{r}_0$ and $\bar{P}$ in the following problem. \vspace{-5pt}
\begin{align*}\mathrm{({\rm P2-NoET})}: ~ \mathop{\mathtt{Maximize}}_{\mv{v}_0,\{\mv{w}_i\}}
& ~~~  0  \\
\mathtt {Subject \ to} & ~~~ r_0 \geq \bar{r}_0 \\
& ~~~ \|\mv{v}_0\|^2+\sum\limits_{i=1}^d\|\mv{w}_i\|^2 \leq \bar{P}.
\end{align*}Note that this problem corresponds to the case where no energy transfer is required, and thus $\mv{w}_i$'s play the only role of AN. The feasibility problem (P2-NoET) can be easily solved by checking whether $\bar{r}_0$ is no larger than the optimal value of (P1-NoET) in Section \ref{sec:Optimal Solution to Problem (P1)}. Without loss of generality, in the rest of this paper, we assume that (P2) is feasible.

Next, we consider another special case of (P2) where no information transmission is required to the IR, i.e., $\bar{r}_0=0$. In this case, $\mv{v}_0=\mv{0}$ and thus (P2) reduces to \vspace{-4pt} \begin{align*}\mathrm{({\rm P2-NoIT})}:~\mathop{\mathtt{Maximize}}_{\{\mv{w}_i\}}
& ~~~  \sum\limits_{k=1}^K \mu_k\zeta \left(\sum\limits_{i=1}^d|\mv{w}_i^H\mv{g}_k|^2\right)  \\
\mathtt {Subject \ to} & ~~~
\sum\limits_{i=1}^d\|\mv{w}_i\|^2 \leq \bar{P}.
\end{align*}

Let $\psi$ and $\mv{\eta}$ denote the maximum eigenvalue and its corresponding unit-norm eigenvector of the matrix $\sum_{k=1}^K\mu_k\zeta\mv{g}_k\mv{g}_k^H$, respectively. From \cite{Rui11}, the optimal value of problem (P2-NoIT) is known to be\begin{align}\label{eqn:max energy}
\begin{small}E_{{\rm max}}=\psi\bar{P},\vspace{-5pt}\end{small}\end{align}which is achieved by $\mv{w}_i^\ast=\sqrt{p_i}\mv{\eta}$, $1\leq i \leq d$, for any set of $p_i$'s satisfying $p_i\geq 0$, $\forall i$, and $\sum_{i=1}^dp_i=\bar{P}$. In other words, the optimal solution to problem (P2-NoIT) is to align all the energy beams to the same direction as $\mv{\eta}$, a technique known as ``energy beamforming'' \cite{Rui11}.\vspace{-5pt}


\subsection{Optimal Solution to (P2)}\label{sec:Optimal Solution to Problem (P2)}

In this subsection, we propose the optimal solution to (P2). Similar to (P1), (P2) can be reformulated into two sub-problems shown in the sequel. First, similar to (P1), it can be shown that there always exists a SINR constraint $\gamma_0>0$ at IR such that the following problem
\begin{align*}\mathrm{(P2.1)}: \\ ~~~ \mathop{\mathtt{Maximize}}_{\mv{v}_0,\{\mv{w}_i\}}
& ~~~ \sum\limits_{k=1}^K \mu_k\zeta \left(|\mv{v}_0^H\mv{g}_k|^2+\sum\limits_{i=1}^d|\mv{w}_i^H\mv{g}_k|^2\right) \\
\mathtt {Subject \ to} & ~~~ \frac{|\mv{v}_0^H\mv{h}|^2}{\sum\limits_{i=1}^d|\mv{w}_i^H\mv{h}|^2+\sigma_0^2}\geq \gamma_0, \\ & ~~~ \frac{|\mv{v}_0^H\mv{g}_k|^2}{\sum\limits_{i=1}^d|\mv{w}_i^H\mv{g}_k|^2+\sigma_k^2}\leq \frac{1+\gamma_0}{2^{\bar{r}_0}}-1, \ \forall k,  \\ & ~~~
\|\mv{v}_0\|^2+\sum\limits_{i=1}^d\|\mv{w}_i\|^2 \leq \bar{P},
\end{align*}has the same optimal solution to (P2). Furthermore, let $g_2(\gamma_0)$ denote the optimal value of problem (P2.1) with a given $\gamma_0>0$, then the optimal value of problem (P2) is the same as that of the following problem
\begin{align*}\mathrm{(P2.2)}:~\mathop{\mathtt{Maximize}}_{\gamma_0>0}
& ~~~ g_2(\gamma_0).
\end{align*}

%
%
%

Let $\gamma_0^\ast$ denote the optimal solution to problem (P2.2). From the above results, with $\gamma_0=\gamma_0^\ast$, it follows that problems (P2) and (P2.1) have the same optimal solution. Therefore, similar to (P1), problem (P2) can be solved in the following two steps: First, given any $\gamma_0>0$, we solve problem (P2.1) to find $g_2(\gamma_0)$; then, we solve problem (P2.2) to obtain the optimal $\gamma_0^\ast$ by a one-dimension search over $\gamma_0>0$. In the rest of this subsection, we focus on solving (P2.1), which is non-convex.

Define $\mv{S}=\mv{v}_0\mv{v}_0^H$ and $\mv{Q}=\sum_{i=1}^d\mv{w}_i\mv{w}_i^H$. Then the SDR of (P2.1) can be expressed as
\begin{align}\mathrm{(P2.1-SDR)}: \nonumber \\ \mathop{\mathtt{Maximize}}_{\mv{S},\mv{Q}}
& ~~~ \sum\limits_{k=1}^K \mu_k\zeta\left({\rm Tr}(\mv{G}_k\mv{S})+{\rm Tr}(\mv{G}_k\mv{Q})\right) \nonumber \\
\mathtt {Subject \ to} & ~~~ {\rm Tr}(\mv{H}\mv{S})\geq \gamma_0 \left({\rm Tr}(\mv{H}\mv{Q})+\sigma_0^2\right), \label{eqn:constraint1}\\ & ~~~ \frac{{\rm Tr}(\mv{G}_k\mv{S})}{\gamma_e} \leq {\rm Tr}(\mv{G}_k\mv{Q})+\sigma_k^2, \ \forall k,  \label{eqn:constraint2}\\ & ~~~
{\rm Tr}(\mv{S})+{\rm Tr}(\mv{Q}) \leq \bar{P}, \label{eqn:constraint3}\\ & ~~~ \mv{S}\succeq \mv{0}, ~~~ \mv{Q}\succeq \mv{0}, \label{eqn:constraint4}
\end{align}where $\mv{H}=\mv{h}\mv{h}^H$, $\mv{G}_k=\mv{g}_k\mv{g}_k^H$, and $\gamma_e=(1+\gamma_0)/2^{\bar{r}_0}-1$. Similar to (P1.1), if the optimal solution to problem (P2.1-SDR), denoted by $\mv{S}^\ast$ and $\mv{Q}^\ast$, satisfies ${\rm rank}(\mv{S}^\ast)= 1$, then the optimal information beam $\mv{v}_0^\ast$ and energy beam $\mv{w}_i^\ast$'s, $i=1,\cdots,d$ ($d={\rm rank}(\mv{Q}^\ast)$), for problem (P2.1) can be obtained from the EVDs of $\mv{S}^\ast$ and $\mv{Q}^\ast$, respectively; otherwise if ${\rm rank}(\mv{S}^\ast)>1$, the optimal value of problem (P2.1-SDR) only serves as an upper bound on that of problem (P2.1). In the following, we show that there always exists an optimal solution with ${\rm rank}(\mv{S}^\ast)=1$ for (P2.1-SDR).

Since (P2.1-SDR) is convex, it can be solved by CVX. Suppose that the resulting optimal solution $(\mv{S}^\ast,\mv{Q}^\ast)$ satisfies ${\rm rank}(\mv{S}^\ast)>1$. Let $E^\ast$ denote the optimal value of (P2.1-SDR) achieved by $(\mv{S}^\ast,\mv{Q}^\ast)$. Then consider the following problem.\begin{align*}\mathrm{(P2.1-SDR-New)}: \ \ \ \ \ \ \ \ \ \ \ \ \ \ \ \ \ \ \ \ \ \ \ \ \ \ \ \ \ \ \ \ \ \ \end{align*} \vspace{-19pt} \begin{align}\mathop{\mathtt{Maximize}}_{\mv{S},\mv{Q}}
& ~~~ \frac{{\rm Tr}(\mv{H}\mv{S})}{{\rm Tr} (\mv{H}\mv{Q})+\sigma_0^2} \nonumber \\
\mathtt {Subject \ to} & ~~~ (\ref{eqn:constraint2}), ~ (\ref{eqn:constraint3}), ~ (\ref{eqn:constraint4}), \label{21} \\ & ~~~ \sum\limits_{k=1}^K\mu_k\zeta({\rm Tr}(\mv{G}_k\mv{S})+{\rm Tr}(\mv{G}_k\mv{Q})) \geq E^\ast. \label{22}
\end{align}It can be shown that with $(\mv{S}^\ast,\mv{Q}^\ast)$, the resulting value of (P2.1-SDR-New) is $\gamma_0$. Let $(\bar{\mv{S}}^\ast,\bar{\mv{Q}}^\ast)$ denote the optimal solution to (P2.1-SDR-New), and $\bar{\gamma}_0$ be the optimal value. Then we have $\bar{\gamma}_0\geq  \gamma_0$. As a result, with the new solution $(\bar{\mv{S}}^\ast,\bar{\mv{Q}}^\ast)$, all the constraints in (P2.1-SDR), i.e, (\ref{eqn:constraint1})--(\ref{eqn:constraint4}), are satisfied, and the optimal value $E^\ast$ is still achieved. Therefore, the optimal solution to (P2.1-SDR-New) is also optimal to (P2.1-SDR). Furthermore, similar to (P1.1-SDR) (see Proposition \ref{proposition4}), it can be shown that there always exists a rank-one optimal covariance solution for $\mv{S}$ to (P2.1-SDR-New). Therefore, we can conclude that there always exists an optimal solution $(\mv{S}^\ast,\mv{Q}^\ast)$ to (P2.1-SDR) with ${\rm rank}(\mv{S}^\ast)=1$, and there is no loss of optimality for (P2.1) due to the rank relaxation on $\mv{S}$ in (P2.1-SDR).


At last, similar to Proposition \ref{proposition3} and Corollaries \ref{corollary1} and \ref{corollary2} in Section \ref{sec:Optimal Solution to Problem (P1)}, it can be shown that if $|\bar{\Psi}|\geq \min(M-1,K)$, where $\bar{\Psi}$ is still given in (\ref{psi2}) but with $\beta_k^\ast$'s denoting the optimal dual solution to (P2.1-SDR) corresponding to (\ref{eqn:constraint2}), ${\rm rank}(\mv{S}^\ast)=1$ is always true for (P2.1-SDR).

\begin{remark}\label{remark2}
It is worth noting that in \cite{Rui13}, a similar problem to (P2) has been studied without considering the secret information transmission to IRs, which in the case of one single IR under the same setup of this paper is equivalent to the following simplified problem of (P2).\begin{align*}\mathrm{({\rm P2-NoSC})}: \\ \mathop{\mathtt{Maximize}}_{\mv{v}_0,\{\mv{w}_i\}}
& ~~~  \sum\limits_{k=1}^K \mu_k\zeta \left(|\mv{v}_0^H\mv{g}_k|^2+\sum\limits_{i=1}^d|\mv{w}_i^H\mv{g}_k|^2\right)  \\
\mathtt {Subject \ to} & ~~~ \log_2\left(1+\frac{|\mv{v}_0^H\mv{h}|^2}{\sum\limits_{i=1}^d|\mv{w}_i^H\mv{h}|^2+\sigma_0^2}\right) \geq \tilde{r}_0 , \\ & ~~~
\|\mv{v}_0\|^2+\sum\limits_{i=1}^d\|\mv{w}_i\|^2 \leq \bar{P},
\end{align*}where $\tilde{r}_0$ denotes the given rate constraint for IR (without secrecy consideration). Note that an important result shown in \cite[Proposition 3.1]{Rui13} is that under the assumption of independent user channels, the optimal solution to problem (P2-NoSC) should satisfy that $\mv{w}_i^\ast=\mv{0}$, $\forall $ $1\leq i \leq d$, i.e., no energy beam is needed, while only the information beam $\mv{v}_0$ is adjusted for achieving the information rate target for IR and yet maximizing the weighted sum-energy transferred to ERs. However, with the newly introduced secrecy rate constraint in (P2), energy beams are in general needed in the optimal solution for the SWIPT system with arbitrary user channels, since they carry AN to reduce the information rate eavesdropped by ERs, especially when ERs have better channels than IR from the transmitter.

\end{remark}

\subsection{Suboptimal Solutions to (P2)}\label{sec:Suboptimal Solutions to Problem (P2)}
Similarly as for (P1), in this subsection, we propose two suboptimal solutions for (P2), which can be designed with lower complexity. Similar to the two suboptimal solutions proposed in Section \ref{sec:Suboptimal Solutions to Problem (P1)} for (P1), in the following we assume that the energy beams $\mv{w}_i$ ($i=1,\cdots,d$) in (P2) are all in the null space of the IR's channel $\mv{h}$. Furthermore, the information beam $\mv{v}_0$ is aligned to the null space of the ERs' channels $\mv{G}=[\mv{g}_1,\cdots,\mv{g}_K]^H$ in the first suboptimal solution, while it is in the same direction as $\mv{h}$ for the second suboptimal solution. Again, the first suboptimal solution is only applicable when $K<M$. In the following, we present the two suboptimal solutions in more details.

\subsubsection{Suboptimal Solution \uppercase\expandafter{\romannumeral1}}

\ \ \ \ \ \ \ \

Supposing that $K<M$, then the first suboptimal solution aims to solve problem (P2) with the additional constraints: $\mv{v}_0^H\mv{g}_k=0$, $\forall k$, and $\mv{w}_i^H\mv{h}=0$, $\forall i$. To satisfy the above constraints, $\mv{v}_0$ and $\mv{w}_i$'s should be in the form of (\ref{eqn:sub v0}) and (\ref{eqn:sub wi}), respectively. Furthermore, with $\tilde{\mv{v}}_0^\ast=\tilde{\mv{V}}^H\mv{h}/\|\tilde{\mv{V}}^H\mv{h}\|$, the secrecy rate of the IR under this scheme is given in (\ref{eqn:sub rate1}). Since all ERs cannot harvest energy from the information beam, to maximize the weighted sum-energy transferred to ERs, $\tilde{P}_0$ should be set to the smallest power to make $r_0^{({\rm I})}=\bar{r}_0$. It thus follows \begin{align}\tilde{P}_0^\ast=\frac{(2^{\bar{r}_0}-1)\sigma_0^2}{\|\tilde{\mv{V}}^H\mv{h}\|^2}.\end{align}To summarize, in this suboptimal solution, we have
\begin{align}\label{eqn:sub1 v0}
\begin{small}\mv{v}_0^\ast=\sqrt{\tilde{P}_0^\ast}\tilde{\mv{V}}\tilde{\mv{v}}_0^\ast=\frac{\sqrt{(2^{\bar{r}_0}-1)\sigma_0^2}\tilde{\mv{V}}\tilde{\mv{V}}^H\mv{h}}{\|\tilde{\mv{V}}^H\mv{h}\|^2}. \end{small}
\end{align}

Notice that the harvested energy of ${\rm ER}_k$ under this suboptimal solution is in the form of (\ref{eqn:sub1 energy}). Thus, to find the optimal $\tilde{\mv{w}}_i^\ast$'s, we need to solve the following problem.
\begin{align*}\mathrm{(P2-Sub1)}:~\mathop{\mathtt{Maximize}}_{\{\tilde{\mv{w}}_i\}}
& ~~~  \sum\limits_{k=1}^K \mu_k\zeta \left(\sum\limits_{i=1}^d\tilde{\mv{w}}_i^H\tilde{\mv{G}}_k\tilde{\mv{w}}_i\right)  \\
\mathtt {Subject \ to} & ~~~
\sum\limits_{i=1}^d\|\tilde{\mv{w}}_i\|^2 \leq \bar{P}-\tilde{P}_0^\ast.
\end{align*}Note that in the above, we have assumed $\bar{P}\geq \tilde{P}^\ast_0$. Let $\tilde{\psi}$ and $\tilde{\mv{\eta}}$ denote the maximum eigenvalue and its corresponding unit-norm eigenvector of the matrix $\sum_{k=1}^K\mu_k\zeta\tilde{\mv{G}}_k$, respectively. Similar to problem (P2-NoIT), it can be shown that the optimal value of problem (P2-Sub1) is $\tilde{E}_{{\rm max}}=\tilde{\psi}(\bar{P}-\tilde{P}_0^\ast)$, which is achieved by $\tilde{\mv{w}}_i^\ast=\sqrt{\tilde{p}_i}\tilde{\mv{\eta}}$, $1\leq i \leq d$, for any set of $\tilde{p}_i$'s satisfying $\sum_{i=1}^d\tilde{p}_i=\bar{P}-\tilde{P}_0^\ast$. In practice, it is preferable to send only one energy beam to minimize the complexity of beamforming implementation at the transmitter; thus, we have
\begin{align}\label{eqn:sub1 wi}
\begin{small}\mv{w}_i^\ast=\left\{\begin{array}{ll}\sqrt{\bar{P}-\tilde{P}_0^\ast}\tilde{\mv{X}}\tilde{\mv{\eta}}, & {\rm if} \ i=1, \\ 0, & {\rm otherwise}.\end{array}\right. \end{small}
\end{align}Note that unlike Suboptimal Solution \uppercase\expandafter{\romannumeral1} for (P1) shown in (\ref{eqn:energy beam 2}), one single energy beam is sufficient in this case.

\subsubsection{Suboptimal Solution \uppercase\expandafter{\romannumeral2}}

\ \ \ \ \ \ \

The second suboptimal solution aims to solve problem (P2) with the additional constraints: $\mv{v}_0=\sqrt{\hat{P}_0}\mv{h}/\|\mv{h}\|$ and $\mv{w}_i^H\mv{h}=0$, $\forall i$, where $\hat{P}_0=\|\mv{v}_0\|^2$ denotes the transmit power of the information beam. Similar to (\ref{eqn:sub1 wi}), it can be shown that the optimal energy beams should be in the following form:
\begin{align}
\begin{small}\mv{w}_i=\left\{\begin{array}{ll}\sqrt{\bar{P}-\hat{P}_0}\tilde{\mv{X}}\tilde{\mv{\eta}}, & {\rm if} \ i=1, \\ 0, & {\rm otherwise}.\end{array}\right. \end{small}
\end{align}Next, we derive the optimal power allocation for $\hat{P}_0$, denoted by $\hat{P}_0^\ast$. It can be shown that the secrecy rate of IR in this scheme is given by
\begin{align}
& \begin{small}r_0^{({\rm II})}=\min_{1\leq k \leq K} ~  \log_2\left(1+\frac{\hat{P}_0\|\mv{h}\|^2}{\sigma_0^2}\right)-\nonumber \end{small} \\ & \begin{small}\log_2\left(1+\frac{\hat{P}_0|\mv{h}^H\mv{g}_k|^2}{\|\mv{h}\|^2((\bar{P}-\hat{P}_0)|\tilde{\mv{\eta}}^H\tilde{\mv{X}}^H\mv{g}_k|^2+\sigma_k^2)}\right).\end{small} \label{eqn:sub rate2}
\end{align}Define the set of feasible power allocation as $\hat{\mathcal{P}}_0=\{\hat{P}_0|r_0^{({\rm II})}\geq \bar{r}_0, 0<\hat{P}_0\leq \bar{P}\}$, which is assumed to be non-empty. To maximize the weighted sum-energy transferred to ERs subject to the secrecy rate constraint of the IR, we need to solve the following power allocation problem.
\begin{align*}\mathrm{(P2-Sub2)}:\\ \mathop{\mathtt{Maximize}}_{\hat{P}_0}
& ~  \begin{small}\frac{\sum\limits_{k=1}^K\mu_k\zeta\hat{P}_0|\mv{h}^H\mv{g}_k|^2}{\|\mv{h}\|^2}+\sum\limits_{k=1}^K\mu_k\zeta(\bar{P}-\hat{P}_0)|\tilde{\mv{\eta}}^H\tilde{\mv{X}}^H\mv{g}_k|^2  \end{small} \\
\mathtt {Subject \ to} & ~
\hat{P}_0\in \hat{\mathcal{P}}_0.
\end{align*}

Let $\hat{P}_0^{{\rm min}}$ and $\hat{P}_0^{{\rm max}}$ denote the minimal and maximal elements in the set $\hat{\mathcal{P}}_0$, respectively. Then it can be shown that the optimal power allocation to (P2-Sub2) is given by
\begin{align}
\begin{small}\hat{P}_0^\ast=\left\{\begin{array}{ll}\hat{P}_0^{{\rm max}}, & {\rm if} \ \frac{\sum\limits_{k=1}^K\mu_k|\mv{h}^H\mv{g}_k|^2}{\|\mv{h}\|^2}\geq \sum\limits_{k=1}^K\mu_k|\tilde{\mv{\eta}}^H\tilde{\mv{X}}^H\mv{g}_k|^2, \\ \hat{P}_0^{{\rm min}}, & {\rm otherwise}.\end{array}\right. \end{small}
\end{align}

\section{Numerical Example}\label{sec:Numerical Results}

In this section, we provide numerical examples to validate our results. In the first numerical example, we consider a MISO SWIPT system in which Tx is equipped with $M=4$ antennas, and there are $K=3$ ERs.\footnote{Note that $K<M$ in this example; thus, Suboptimal Solution \uppercase\expandafter{\romannumeral1} for (P1) or (P2) is feasible.} We assume that the signal attenuation from Tx to all ERs is $30$dB corresponding to an equal distance of $1$ meter, i.e., $\rho_{g_k}^2=-30$dB, $1\leq k \leq K$, and that from Tx to the IR is $70$dB corresponding to a distance of $20$ meters, i.e., $\rho_h^2=-70$dB. The channel vectors $\mv{g}_k$'s and $\mv{h}$ are randomly generated from i.i.d. Rayleigh fading with the respective average power values specified as above. We set $\bar{P}=1$Watt (W) or $30$dBm, $\zeta=50\%$, and $\sigma_k^2=-50$dBm, $0\leq k \leq K$. We also set $\mu_k=1$, $1\leq k \leq K$ in (P2); thus, the sum-energy harvested by all ERs is considered.


\begin{figure}
\begin{center}
\subfigure[$\log_2\left(\frac{1+g_1(\gamma_e)}{1+\gamma_e}\right)$ versus $\gamma_e$.]{\scalebox{0.52}{\includegraphics*{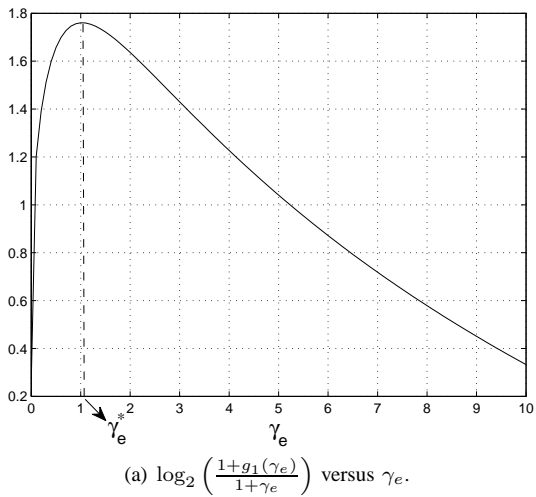}}}
\subfigure[$g_2(\gamma_0)$ versus $\gamma_0$.]{\scalebox{0.52}{\includegraphics*{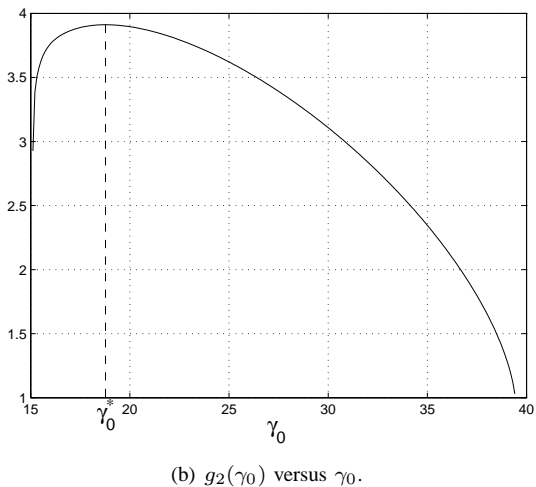}}}
\caption{Uniqueness of $\gamma_e^\ast$ in (P1.2) and $\gamma_0^\ast$ in (P2.2).}\label{fig3}
\end{center}\vspace{-15pt}

\end{figure}

First, we illustrate the two-stage optimization approach to solve (P1) and (P2), as proposed in Section \ref{sec:Optimal Solution to Problem (P1)} and Section \ref{sec:Optimal Solution to Problem (P2)}, respectively. Figs. \ref{fig3}(a) and \ref{fig3}(b) show the plot of $\log_2\left(\frac{1+g_1(\gamma_e)}{1+\gamma_e}\right)$ over $\gamma_e>0$ in (P1) with the individual harvested energy constraints $\bar{E}_k$'s of ERs set as $1$mW, $\forall k$, and the plot of $g_2(\gamma_0)$ over $\gamma_0>0$ in (P2) with the secrecy rate constraint of IR $\bar{r}_0$ set as $4$ bits per second (bps) per Hz, respectively. It is observed that in this particular setup (and many others used in our simulations for which the results are not shown here due to the space limitation) there is only one single maximum point in each of the two plotted functions; however, we are unable yet to verify analytically the concavity or even the quasi-concavity of these two functions.

Next, we adopt the Rate-Energy (R-E) region \cite{Liang13}, which consists of all the achievable (secrecy) rate and harvested energy pairs for a given sum-power constraint $\bar{P}$, to compare the performances of the optimal and suboptimal solutions for (P1) proposed in Section \ref{sec:Proposed Solutions to Problem (P1)}. Note that in general the R-E region in our setup is a $(K+1)$-dimension region given one IR and $K$ ERs. For simplicity, in the following we assume that all ERs have identical energy constraints, denoted by $E\geq 0$; thus, the R-E region reduces to a two-dimension region, which is given by
\begin{align}\begin{small}\mathcal{C}_{{\rm R-E \ (P1)}}\triangleq \bigcup\limits_{\|\mv{v}_0\|^2+\sum\limits_{i=1}^d\|\mv{w}_i\|^2\leq \bar{P}} \bigg\{(R,E): R \leq r_0, E\leq
E_k, \forall k\bigg\},\end{small}\label{eqn:R-E region (P1)}\end{align}where $r_0$ and $E_k$ are given in (\ref{eqn:secrecy rate}) and (\ref{eqn:harvested energy}), respectively. Note that by solving (P1) with $\bar{E}_k=\bar{E}$, $\forall k$, and by changing the values of $\bar{E}$, we can characterize the boundary of the resulting R-E region defined in (\ref{eqn:R-E region (P1)}).

\begin{figure}
\begin{center}
 \scalebox{0.55}{\includegraphics*{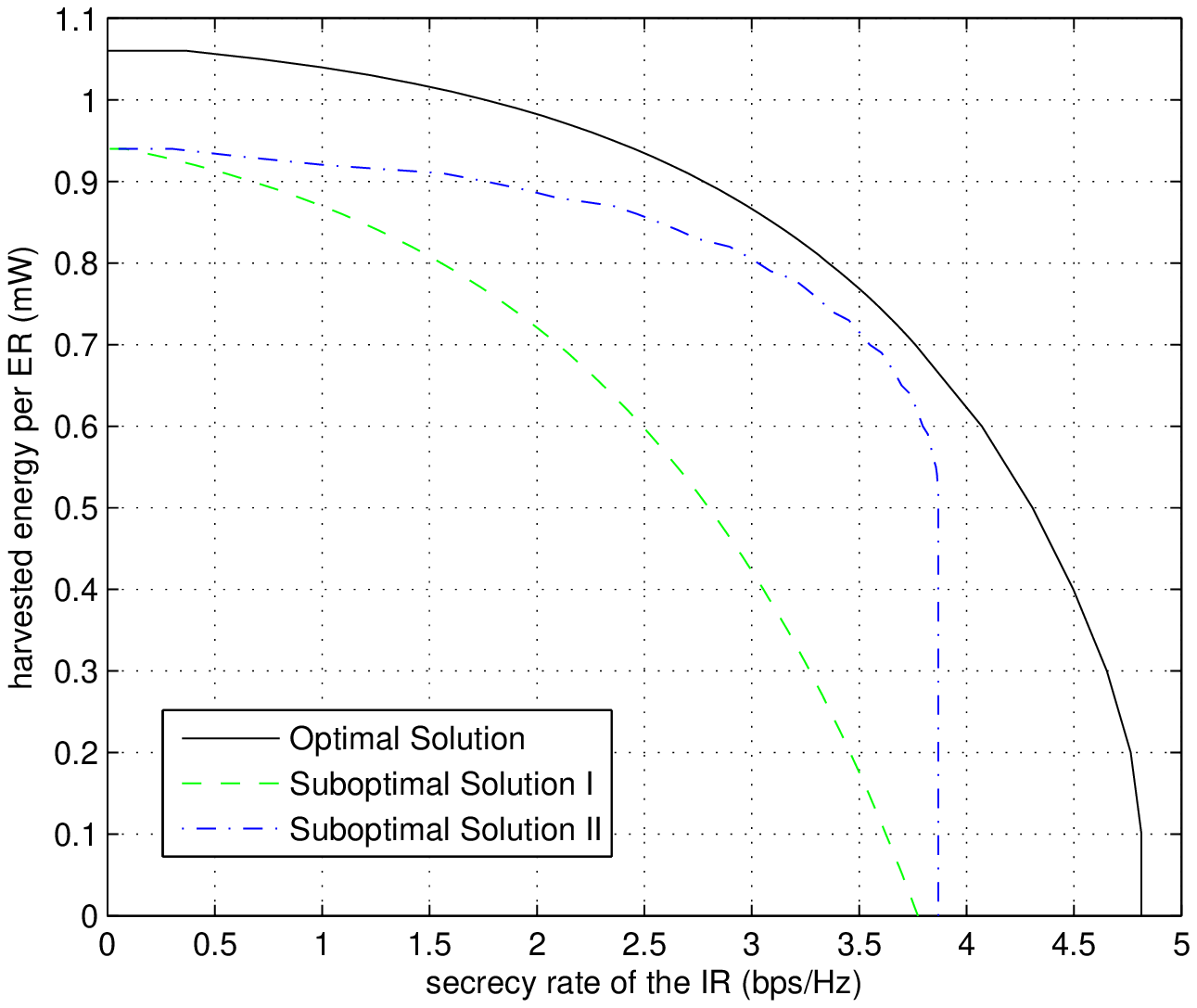}}
 \end{center}
\caption{Achievable R-E region by the proposed solutions for (P1).} \label{fig5}\vspace{-10pt}
\end{figure}

\begin{figure}
\begin{center}
 \scalebox{0.55}{\includegraphics*{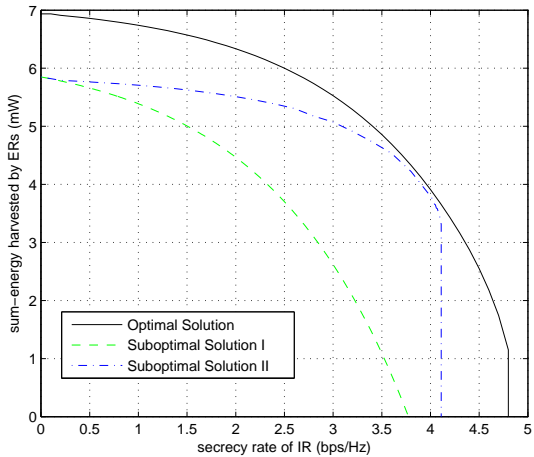}}
 \end{center}
\caption{Achievable R-E region by the proposed solutions for (P2).} \label{fig4}\vspace{-10pt}
\end{figure}

Fig. \ref{fig5} compares the R-E regions achieved by different information and energy beamforming solutions for (P1). It is observed that the optimal solution achieves the best R-E trade-offs. Moreover, Suboptimal Solution \uppercase\expandafter{\romannumeral2} is observed to perform better than Suboptimal Solution \uppercase\expandafter{\romannumeral1}, especially when the achievable secrecy rate for the IR is large. However, it is worth noting that Suboptimal Solution \uppercase\expandafter{\romannumeral1} has the lowest complexity among the three proposed solutions. Notice that for this suboptimal solution, closed-form expressions of the optimal information/energy beamforming vectors and their power allocation are given in (\ref{eqn:infor beam2}) and (\ref{eqn:energy beam 2}), respectively. Furthermore, with no information leakage to ERs with the designed information beamforming, i.e., $\mv{v}_0^H\mv{g}_k=0$, $\forall k$, there is no need to design a special codebook for the secrecy information signal at the transmitter \cite{Liu10}, \cite{Oggier11}.

Next, we compare the performances of the optimal and suboptimal solutions proposed for (P2) in Section \ref{sec:Proposed Solutions to Problem (P2)}. In this case, the R-E region in general consists of all pairs of the achievable (secrecy) rate for IR and the harvested sum-energy for ERs for a given sum-power constraint $\bar{P}$. Specifically, the R-E region is defined as
\begin{align}\begin{small}\mathcal{C}_{{\rm R-E \ (P2)}}\triangleq \bigcup\limits_{\|\mv{v}_0\|^2+\sum\limits_{i=1}^d\|\mv{w}_i\|^2\leq \bar{P}} \bigg\{(R,E): R \leq r_0, E\leq
\sum\limits_{k=1}^KE_k\bigg\}.\end{small}\label{eqn:R-E region}\end{align}Note that by solving problem (P2) with different values of $\bar{r}_0$, we can characterize the boundary of the resulting R-E region defined in (\ref{eqn:R-E region}).

Fig. \ref{fig4} shows three R-E regions achieved by different information and energy beamforming schemes for (P2). It is observed that similar to Fig. \ref{fig5}, the optimal solution achieves the best R-E trade-offs, while Suboptimal Solution \uppercase\expandafter{\romannumeral2} works better than Suboptimal Solution \uppercase\expandafter{\romannumeral1}. From the results in both Figs. \ref{fig5} and \ref{fig4}, it is inferred that in general it is more beneficial to align the information beam $\mv{v}_0$ to the same direction as the IR's channel $\mv{h}$ rather than to the null space of ERs' channels in both (P1) and (P2).

\begin{figure}
\begin{center}
 \scalebox{0.55}{\includegraphics*[479pt,237pt][208pt,524pt]{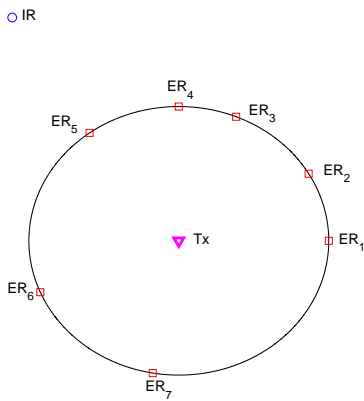}}
 \end{center}
\caption{Locations of the IR and ERs.} \label{fig6}\vspace{-10pt}
\end{figure}

In the second numerical example, we consider a MISO SWIPT system as shown in Fig. \ref{fig6}, where there are $K=7$ ERs and the IR is equipped with $M=9$ antennas. In this example, we use the far-field uniform linear antenna array \cite{Luo07} to model the channels. Specifically,
\begin{align}
& \mv{h}=\rho_{h}\times [1,e^{j\theta_0},\cdots,e^{j(M-1)\theta_0}]^T, \\
& \mv{g}_k=\rho_{g_k}\times [1,e^{j\theta_k},\cdots,e^{j(M-1)\theta_k}]^T, ~~ k=1,\cdots,K,
\end{align}{where $\rho_h^2=-70$dB, $\rho_{g_k}^2=-30$dB, $1\leq k \leq K$, and $\theta_n=-\frac{2\pi d \sin(\phi_n)}{\lambda}$, $n=0,1,\cdots,K$, with $d$ denoting the spacing between successive antenna elements at the Tx, $\lambda$ denoting the carrier wavelength, and $\phi_0$ denoting the direction of the IR to Tx, and $\phi_n$ for that of ${\rm ER}_n$ to Tx, $1\leq n \leq K$. We set $d=\frac{\lambda}{2}$, and $\{\phi_0,\phi_1,\cdots,\phi_7\}=\{\frac{11\pi}{16},0,\frac{\pi}{6},\frac{3\pi}{8},\frac{\pi}{2},\frac{45\pi}{64},\frac{9\pi}{8},\frac{13\pi}{9}\}$. The other parameters are set the same as those in the first numerical example.

\begin{figure}
\begin{center}
 \scalebox{0.55}{\includegraphics*{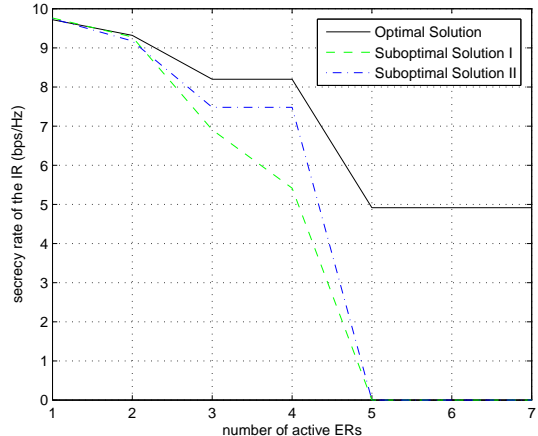}}
 \end{center}
\caption{The secrecy rate of the IR over the number of active ERs with given per-ER energy constraint, $\bar{E}_k=0.8$mW.} \label{fig7}\vspace{-10pt}
\end{figure}

In this example, we activate one more ER at each time (from ${\rm ER}_1$ to ${\rm ER}_7$). Fig. \ref{fig7} shows the secrecy rate achieved by our proposed optimal and suboptimal algorithms for (P1) against the number of active ERs in the system with $\bar{E}_k=0.8$mW, $\forall k$. It is observed that with more ERs (or eavesdroppers) activated, the achievable secrecy rate for the IR is reduced for all proposed algorithms. It is also observed that when ${\rm ER}_5$ is activated, there is a drastic decrease in the secrecy rate achieved for the IR. This is because as shown in Fig. \ref{fig6}, ${\rm ER}_5$ is aligned in a direction very close to that of the IR ($\phi_5\approx \phi_0$) but with higher channel power due to shorter distance from the Tx. Furthermore, it is observed that after ${\rm ER}_5$ is activated, both Suboptimal Solutions I and II achieve zero secrecy rate. The reason is as follows. Note that for both of these two suboptimal solutions, the energy beams $\mv{w}_i$'s are aligned into the null space of $\mv{h}$, i.e., $\mv{w}_i^H\mv{h}=0$, $\forall i$. However, in this example the direction of $\mv{g}_5$ is very close to that of $\mv{h}$. It thus follows that $\mv{w}_i^H\mv{g}_5\approx 0$, $\forall i$. In other words, the energy beams cannot play the role of AN to reduce ${\rm ER}_5$'s SINR in this case. Moreover, since ${\rm ER}_5$ has better channel than the IR, the achievable secrecy rate becomes close to zero.

\begin{figure}
\begin{center}
 \scalebox{0.55}{\includegraphics*{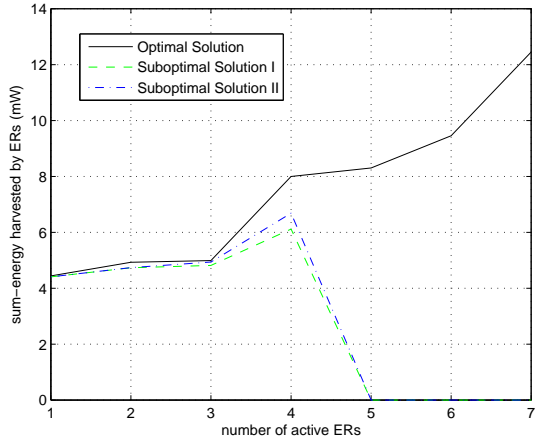}}
 \end{center}
\caption{The sum-energy harvested by ERs over the number of active ERs with given secrecy rate constraint for the IR, $\bar{r}_0=4$bps/Hz.} \label{fig8}\vspace{-10pt}
\end{figure}

Fig. \ref{fig8} shows the sum-energy harvested by all ERs by the proposed optimal and suboptimal algorithms for (P2) against the number of active ERs with $\bar{r}_0=4$bps/Hz. It is observed that with more ERs, the sum-energy harvested is increased in all cases. Furthermore, it is observed that when $K\leq 4$, the performance of both Suboptimal Solutions I and II is very close to that of the optimal solution. However, after ${\rm ER}_5$ is activated, both of the suboptimal solutions achieve zero sum-energy because the secrecy rate constraint cannot be satisfied in (P2) due to the same reason as given for Fig. \ref{fig7}. From the results in Figs. \ref{fig7} and \ref{fig8}, it is inferred that even in the challenging scenario where one ER is aligned in a direction very close to (but not the same as) the IR, our proposed optimal algorithm still achieves good performance thanks to the jointly optimized beamforming and power allocation design. However, in this case both the two suboptimal solutions cannot perform well. \vspace{-10pt}

\section{Conclusion}\label{sec:Conclusion}
This paper is an initial attempt to address the important issue of physical-layer security in an emerging new type of wireless network with simultaneous wireless information and power transfer (SWIPT). Under the MISO setup with one single IR and multiple ERs, the joint information and energy beamforming design is investigated for the first time to maximize the secret information transmission rate to the IR and yet guarantee the target amount of energy transferred to ERs, or vice verse. We propose efficient algorithms to optimally solve the formulated non-convex design problems by applying the technique of semidefinite relaxation (SDR), and show that SDR has no loss of optimality by exploiting the particular structures of the studied problems. Two suboptimal beamforming designs with lower complexity are also presented, and their performances are compared against that of the optimal solution in terms of achievable (secrecy) rate-energy trade-off. Our results reveal interesting new insights to optimally managing the interference in a secrecy SWIPT system since it plays both the roles of an energy-carrying signal for wireless energy transfer as well as an artificial noise (AN) to protect the secrecy information transmission.

\begin{appendix}

\vspace{10pt}

\subsection{Proof of Lemma \ref{lemma7}}\label{appendix6}

First, we show that $\lambda^\ast>0$. Let $\mv{A}_1^\ast$, $\mv{B}_1^\ast$, and $\xi_1^\ast$ be given in (\ref{eqn:A2}), (\ref{eqn:B2}), and (\ref{eqn:pi2}), respectively, by substituting the optimal dual solution of problem (P1.1-SDR-Eqv). To ensure that the Lagrangian in (\ref{eqn:Lagrangian2}) is bounded from above such that the dual function exists, it follows that
\begin{align}\label{eqn:condition1}
\begin{small} \mv{A}_1^\ast\preceq \mv{0}, ~~~ \mv{B}_1^\ast\preceq \mv{0}, ~~~ \xi_1^\ast \leq 0. \end{small}
\end{align}According to (\ref{eqn:Lagrangian2}), the dual problem of (P1.1-SDR-Eqv) can be expressed as
\begin{align*}\mathrm{(P1.1-SDR-Eqv-Dual)}:\ \ \ \ \ \ \ \ \ \ \ \ \ \ \ \ \ \ \ \ \ \ \ \ \ \ \ \ \ \ \ \ \end{align*}\begin{align*}\mathop{\mathtt{Minimize}}_{\lambda,\{\beta_k\},\{\alpha_k\},\theta}
& ~~~  \lambda  \\
\mathtt {Subject \ to} & ~~~
\mv{A}_1\preceq \mv{0}, ~~~ \mv{B}_1\preceq \mv{0}, ~~~ \xi_1 \leq 0,
\\ & ~~~ \beta_k\geq 0, ~~~ \alpha_k\geq 0, ~~~ \forall k, ~~~ \theta\geq 0.
\end{align*}Since the duality gap between (P1.1-SDR-Eqv) and its dual problem (P1.1-SDR-Eqv-Dual) is zero, $\lambda^\ast$ is equal to the optimal value of (P1.1-SDR-Eqv). Therefore, we have $\lambda^\ast>0$.

Next, we show that $\theta^\ast>0$ by contradiction. Define $\phi=\{k|(\beta_k^\ast)^2+(\alpha_k^\ast)^2>0, ~ k=1,\cdots,K\}$. In the following, we discuss two cases in each of which we show that (\ref{eqn:condition1}) cannot be true if $\theta^\ast=0$.

\subsubsection{The case of $\phi=\emptyset$}
Suppose that $\theta^\ast=0$. In this case, we have $\mv{A}_1^\ast=\mv{H}\succeq \mv{0}$, which contradicts to (\ref{eqn:condition1}). Thus, in this case, $\theta^\ast>0$ must be true.

\subsubsection{The case of $\phi\neq \emptyset$}
Suppose that $\theta^\ast=0$. Then in this case, we have $\mv{B}_1^\ast=-\lambda^\ast\mv{H}^\ast+\sum_{k\in \phi}(\beta_k^\ast\gamma_e+\alpha_k^\ast \zeta)\mv{G}_k$. Since $\sum_{k\in \phi}(\beta_k^\ast\gamma_e+\alpha_k^\ast \zeta)\mv{G}_k\succeq \mv{0}$ and $\lambda^\ast>0$, to guarantee that $\mv{B}_1^\ast\preceq \mv{0}$, it requires that any $\mv{x}\in \mathbb{C}^{M\times 1}$ that lies in the null space of $\mv{H}$ must also be in the null space of $\mv{G}_k$, $\forall k\in \phi$; however, this cannot be true since all the channels $\mv{h}$ and $\mv{g}_k$'s are assumed to be linearly independent. Thus, in this case, we also conclude that $\theta^\ast>0$.

By combining the above two cases, it follows that $\theta^\ast>0$. Lemma \ref{lemma7} is thus proved. \vspace{-15pt}

\subsection{Proof of Proposition \ref{proposition4}}\label{appendix7}

The Karush-Kuhn-Tucker (KKT) conditions of problem (P1.1-SDR-Eqv) are expressed as \vspace{-5pt}
\begin{equation}\label{eqn:kkt1}
\begin{small}\mv{A}_1^\ast\mv{S}^\ast=\mv{0}, ~~~ \mv{B}_1^\ast\mv{Q}^\ast=\mv{0}. \end{small}\vspace{-5pt}
\end{equation}

First, we show that ${\rm rank}(\mv{Q}^\ast)\leq \min(K,M)$. The proof directly follows if $K\geq M$ since ${\rm rank}(\mv{Q}^\ast)\leq M=\min(K,M)$. Thus, in the following we focus on the case of $K<M$.

\begin{lemma}\label{lemma8}
Let $\mv{Y}$ and $\mv{Z}$ be two matrices of the same dimension. It then holds that ${\rm rank}(\mv{Y}+\mv{Z})\geq {\rm rank}(\mv{Y})-{\rm rank}(\mv{Z})$.
\end{lemma}

\begin{proof}
It is known that ${\rm rank}(\mv{Y})+{\rm rank}(\mv{Z}) \geq {\rm rank}(\mv{Y}+\mv{Z})$ if $\mv{Y}$ and $\mv{Z}$ are of the same dimension. Then we have ${\rm rank}(\mv{Y+Z})+{\rm rank}(\mv{-Z}) \geq {\rm rank}(\mv{Y})$. Since ${\rm rank}(\mv{Z})={\rm rank}(-\mv{Z})$, Lemma \ref{lemma8} is proved.
\end{proof}

Define $\mv{C}_1^\ast=-\lambda^\ast \mv{H}-\theta^\ast\mv{I}$. Since according to Lemma \ref{lemma7} we have $\lambda^\ast>0$ and $\theta^\ast>0$, it follows that $\mv{C}_1^\ast\prec \mv{0}$ and thus ${\rm rank}(\mv{C}_1^\ast)=M$. Furthermore, $\mv{B}_1^\ast$ can be expressed as $\mv{B}_1^\ast=\mv{C}_1^\ast+\sum_{k=1}^K(\beta_k^\ast\gamma_e+\alpha_k^\ast\zeta)\mv{G}_k$. As a result, according to Lemma \ref{lemma8} we have
\begin{align}
\begin{small}{\rm rank}(\mv{B}_1^\ast)\end{small}& \begin{small}\geq {\rm rank}(\mv{C}_1^\ast)-{\rm rank}\left(\sum\limits_{k=1}^K(\beta_k^\ast\gamma_e+\alpha_k^\ast\zeta)\mv{G}_k\right)\end{small} \nonumber \\ & \begin{small} \overset{(a)}{\geq}M-K, \end{small} \label{eqn:ra2}
\end{align}where $(a)$ is due to the fact that ${\rm rank}\left(\sum_{k=1}^K(\beta_k^\ast\gamma_e+\alpha_k^\ast\zeta)\mv{G}_k\right)\leq K$. According to (\ref{eqn:kkt1}), $\mv{Q}^\ast$ must lie in the null space of $\mv{B}_1^\ast$. Therefore, if $K<M$, ${\rm rank}(\mv{Q}^\ast)\leq M-{\rm rank}(\mv{B}_1^\ast) \leq K$. By combining the above two cases of $K\geq M$ and $K<M$, it follows that ${\rm rank}(\mv{Q}^\ast)\leq \min(K,M)$. The first part of Proposition \ref{proposition4} is thus proved.

Next, we prove the second part of Proposition \ref{proposition4}. Define \begin{align}\begin{small}\mv{D}_1^\ast\end{small} & \begin{small}=-\lambda^\ast\mv{H}-\sum_{k=1}^K\beta_k^\ast\mv{G}_k+\sum_{k=1}^K\alpha_k^\ast\zeta\mv{G}_k-\theta^\ast\mv{I}\end{small} \nonumber \\ & \begin{small}=\mv{B}_1^\ast-\sum\limits_{k=1}^K\left(1+\gamma_e \right)\beta_k^\ast \mv{G}_k. \end{small} \label{eqn:D1} \end{align}Then we have \begin{align}\label{eqn:new A2}\begin{small} \mv{A}_1^\ast=\mv{D}_1^\ast+(1+\lambda^\ast)\mv{H}. \end{small} \end{align}Define $l_1={\rm rank}(\mv{D}_1^\ast)$. If $l_1=M$, then we can conclude that ${\rm rank}(\mv{A}_1^\ast)\geq M-1$ according to (\ref{eqn:new A2}) and Lemma \ref{lemma8}. However, if ${\rm rank}(\mv{A}_1^\ast)=M$, then according to (\ref{eqn:kkt1}) it follows that $\mv{S}^\ast=\mv{0}$, which cannot be the optimal solution to (P1-SDR-Eqv). Therefore, we have ${\rm rank}(\mv{A}_1^\ast)=M-1$ and thus $\mv{S}^\ast=b\mv{\tau}_1\mv{\tau}_1^H$ if $l_1=M$, where $\mv{\tau}_1$ spans the null space of $\mv{A}_1^\ast$. Next, we consider the case where $\mv{D}_1^\ast$ is not full-rank, i.e., $l_1<M$. In this case, let $\mv{\Pi}_1\in \mathbb{C}^{M\times (M-l_1)}$ with $\mv{\Pi}_1^H\mv{\Pi}_1=\mv{I}$ denote the orthogonal basis for the null space of $\mv{D}_1^\ast$, i.e., $\mv{D}_1^\ast\mv{\Pi}_1=\mv{0}$. Let $\mv{\pi}_{1,n}$ denote the $n$th column of $\mv{\Pi}_1$, $1\leq n \leq M-l_1$. Then we have
\begin{align}
\mv{\pi}_{1,n}^H\mv{A}_1^\ast\mv{\pi}_{1,n} & =\mv{\pi}_{1,n}^H\left(\mv{D}_1^\ast+(1+\lambda^\ast)\mv{H}\right)\mv{\pi}_{1,n}\nonumber \\ & =(1+\lambda^\ast)|\mv{h}^H\mv{\pi}_{1,n}|^2, ~ 1\leq n \leq M-l_1.
\end{align}Since $\mv{A}_1^\ast\preceq \mv{0}$ and $1+\lambda^\ast>0$, it follows that $|\mv{h}^H\mv{\pi}_{1,n}|^2=0$, $\forall n$, or \begin{align}\label{eqn:null space of H}\begin{small}\mv{H}\mv{\Pi}_1=\mv{0}. \end{small} \end{align}As a result, we have
\begin{align}\label{eqn:X is optimal}
\begin{small} \mv{A}_1^\ast\mv{\Pi}_1=\left(\mv{D}_1^\ast+(1+\lambda^\ast)\mv{H}\right)\mv{\Pi}_1=\mv{0}. \end{small}
\end{align}Moreover, according to (\ref{eqn:new A2}) and Lemma \ref{lemma8}, we have
\begin{align}
\begin{small}{\rm rank}(\mv{A}_1^\ast)\geq {\rm rank}(\mv{D}_1^\ast)-{\rm rank}\left((1+\lambda^\ast)\mv{H}\right)=l_1-1. \end{small}
\end{align}Let $\mv{\Omega}_1$ denote the orthogonal basis for the null space of $\mv{A}_1^\ast$, it then follows that\begin{align}\label{eqn:rank up}\begin{small}{\rm rank}(\mv{\Omega}_1)=M-{\rm rank}(\mv{A}_1^\ast)\leq M-l_1+1.\end{small}\end{align}Next, we show that ${\rm rank}(\mv{\Omega}_1)=M-l_1+1$. According to (\ref{eqn:X is optimal}), $\mv{\Pi}_1$ spans $M-l_1$ orthogonal dimensions of the null space of $\mv{A}_1^\ast$, i.e., ${\rm rank}(\mv{\Omega}_1)\geq M-l_1$. Suppose that ${\rm rank}(\mv{\Omega}_1)=M-l_1$; then we have $\mv{\Omega}_1=\mv{\Pi}_1$. According to (\ref{eqn:condition1}) and (\ref{eqn:kkt1}), $\mv{S}^\ast$ can be expressed as $\mv{S}^\ast=\sum_{n=1}^{M-l_1} a_n\mv{\pi}_{1,n}\mv{\pi}_{1,n}^H$, where $a_n\geq 0$, $\forall n$. However, in this case, no information is transferred to IR since according to (\ref{eqn:null space of H}), $\mv{\pi}_{1,n}$'s all lie in the null space of $\mv{H}$. As a result, according to (\ref{eqn:rank up}) there exists only one single subspace spanned by $\mv{\tau}_1\in\mathbb{C}^{M\times 1}$ of unit norm, which lies in the null space of $\mv{A}_1^\ast$, i.e., $\mv{A}_1^\ast\mv{\tau}_1=\mv{0}$, and is orthogonal to the span of $\mv{\Pi}_1$, i.e., $\mv{\Pi}_1^H\mv{\tau}_1=\mv{0}$. To summarize, we have
\begin{align}
\begin{small}\mv{\Omega}_1=[\mv{\Pi}_1 \ \mv{\tau}_1], \end{small}
\end{align}and thus ${\rm rank}(\mv{\Omega}_1)=M-l_1+1$. Moreover, according to (\ref{eqn:condition1}) and (\ref{eqn:kkt1}), any optimal solution $\mv{S}^\ast$ to problem (P1.1-SDR-Eqv) can be expressed as $\mv{S}^\ast=\sum_{n=1}^{M-l_1}a_n\mv{\pi}_{1,n}\mv{\pi}_{1,n}^H+b\mv{\tau}_1\mv{\tau}_1^H$, where $a_n\geq 0$, $\forall n$, and $b>0$. The second part of Proposition \ref{proposition4} is thus proved.

Last, we prove the third part of Proposition \ref{proposition4}. Suppose that $(\mv{S}^\ast,\mv{Q}^\ast,t^\ast)$ is an optimal solution to problem (P1.1-SDR-Eqv), where $\mv{S}^\ast$ is given in (\ref{eqn:feasible rank}) and ${\rm rank}(\mv{S}^\ast)>1$. Then consider the new solution $(\bar{\mv{S}}^\ast,\bar{\mv{Q}}^\ast,\bar{t}^\ast)$ given in (\ref{eqn:new S})-(\ref{eqn:new t}). It can be shown that with this new solution we have
\begin{align}
& \begin{small}{\rm Tr}(\mv{H}\bar{\mv{S}}^\ast)={\rm Tr}\left(\mv{H}\left(\mv{S}^\ast-\sum\limits_{n=1}^{M-l_1}a_n\mv{\pi}_{1,n}\mv{\pi}_{1,n}^H\right)\right) ={\rm Tr}(\mv{H}\mv{S}^\ast), \end{small}\label{eqn1} \\
& \begin{small}{\rm Tr}(\mv{H}\bar{\mv{Q}}^\ast)+\bar{t}^\ast\sigma_0^2={\rm Tr}\left(\mv{H}\left(\mv{Q}^\ast+\sum\limits_{n=1}^{M-l_1}a_n\mv{\pi}_{1,n}\mv{\pi}_{1,n}^H\right)\right) +t^\ast \sigma_0^2 \end{small} \nonumber \\ & \ \ \ \ \ \ \ \ \ \ \ \ \ \ \ \ \ \ \ \begin{small}={\rm Tr}(\mv{H}\mv{Q}^\ast)+t^\ast\sigma_0^2=1, \end{small} \label{eqn2} \\
& \begin{small}{\rm Tr}(\mv{G}_k\bar{\mv{S}}^\ast)\leq {\rm Tr}(\mv{G}_k\mv{S}^\ast) \leq \gamma_e({\rm Tr}(\mv{G}_k\mv{Q}^\ast)+t^\ast\sigma_k^2) \end{small} \nonumber \\ & \ \ \ \ \ \ \ \ \ \ \ \ \begin{small}\leq \gamma_e({\rm Tr}(\mv{G}_k\bar{\mv{Q}}^\ast)+\bar{t}^\ast \sigma_k^2), ~~~ \forall k, \end{small} \label{eqn3} \\
& \begin{small} \zeta({\rm Tr}(\mv{G}_k\bar{\mv{S}}^\ast)+{\rm Tr}(\mv{G}_k\bar{\mv{Q}}^\ast))=\zeta({\rm Tr}(\mv{G}_k\mv{S}^\ast)+{\rm Tr}(\mv{G}_k\mv{Q}^\ast))\end{small} \nonumber \\ & \ \ \ \ \ \ \ \ \ \ \ \ \ \ \ \ \ \ \ \ \ \ \ \ \ \ \ \ \ \ \begin{small} \geq \bar{t}^\ast \bar{E}_k, ~~~ \forall k, \end{small}\label{eqn4} \\
& \begin{small} {\rm Tr}(\bar{\mv{S}}^\ast)+{\rm Tr}(\bar{\mv{Q}}^\ast)={\rm Tr}(\mv{S}^\ast)+{\rm Tr}(\mv{Q}^\ast)\leq \bar{t}^\ast \bar{P}, \end{small} \label{eqn5} \\
& \begin{small} \bar{\mv{S}}^\ast \succeq \mv{0}, ~~~ \bar{\mv{Q}}^\ast \succeq \mv{0}, ~~~ \bar{t}^\ast>0. \end{small} \label{eqn6}
\end{align}(\ref{eqn1}) indicates that the new solution $(\bar{\mv{S}}^\ast,\bar{\mv{Q}}^\ast,t^\ast)$ can achieve the same optimal value of (P1.1-SDR-Eqv), while (\ref{eqn2})-(\ref{eqn6}) imply that the new solution satisfies all the constraints of (P1.1-SDR-Eqv). Thus, $(\bar{\mv{S}}^\ast,\bar{\mv{Q}}^\ast,t^\ast)$ is also an optimal solution to (P1.1-SDR-Eqv), with ${\rm rank}(\bar{\mv{S}}^\ast)=1$.

Proposition \ref{proposition4} is thus proved. \vspace{-15pt}

\subsection{Proof of Proposition \ref{proposition3}}\label{appendix5}

According to Proposition \ref{proposition4}, if $l_1={\rm rank}(\mv{D}_1^\ast)=M$, then $\mv{S}^\ast=b\mv{\tau}_1\mv{\tau}_1^H$, and it thus follows that ${\rm rank}(\mv{S}^\ast)=1$ is always true for (P1.1-SDR-Eqv). Moreover, since $\mv{B}_1^\ast \preceq \mv{0}$ and $-\sum_{k=1}^K(1+\gamma_e)\beta_k^\ast\mv{G}_k\preceq \mv{0}$, we have $\mv{D}_1^\ast\preceq \mv{0}$ according to (\ref{eqn:D1}). As a result, to show ${\rm rank}(\mv{D}_1^\ast)=M$, it is sufficient to verify that the maximum eigenvalue of $\mv{D}_1^\ast$ is negative, i.e., $\mv{D}_1^\ast\prec \mv{0}$. Therefore, in the following we show by contradict that if there is no non-zero solution to the equations given in (\ref{eqn:sufficient condition}), then the maximum eigenvalue of $\mv{D}_1^\ast$ must be negative.

Since $\mv{D}_1^\ast\preceq \mv{0}$, its maximum eigenvalue can be either zero or negative. Suppose that the maximum eigenvalue of $\mv{D}_1^\ast$ is zero. Then there exists at least an $\mv{x}\in \mathbb{C}^{M\times 1}\neq \mv{0}$ such that $\mv{x}^H\mv{D}_1^\ast\mv{x}=0$. Since $\mv{B}_1^\ast \preceq \mv{0}$ and $-\sum_{k=1}^K(1+\gamma_e)\beta_k^\ast\mv{G}_k\preceq \mv{0}$, according to (\ref{eqn:D1}) we have \vspace{-7pt}
\begin{align}
& \begin{small}\mv{x}^H\mv{B}_1^\ast\mv{x}=0, \end{small}\label{eqn:1} \\
& \begin{small}\mv{x}^H\sum\limits_{k=1}^K\left(1+\gamma_e\right)\beta_k^\ast\mv{G}_k\mv{x}=0. \end{small}\label{eqn:2}
\end{align}From (\ref{eqn:2}), we have \vspace{-5pt} \begin{align}\label{eqn:tight}\begin{small}\mv{x}^H\mv{G}_k\mv{x}=0, ~~~ {\rm if} \ k\in \bar{\Psi}, \end{small}\end{align}where $\bar{\Psi}$ is given in (\ref{psi2}). Note that (\ref{eqn:tight}) is equivalent to $\mv{G}_k\mv{x}=\mv{0}$, $\forall k\in \bar{\Psi}$, since $\mv{G}_k\succeq \mv{0}$. Moreover, since $\mv{A}_1^\ast \preceq \mv{0}$ and $\lambda^\ast>0$ according to Lemma \ref{lemma7}, it follows from (\ref{eqn:new A2}) that \vspace{-15pt}
\begin{align}\begin{small}\label{eqn:null space}\mv{x}^H\mv{H}\mv{x}=0.\end{small}\end{align}Note that (\ref{eqn:null space}) is equivalent to $\mv{H}\mv{x}=\mv{0}$ since $\mv{H}\succeq \mv{0}$. Thus, from (\ref{eqn:1})-(\ref{eqn:null space}), we have
\begin{align}
\begin{small}\mv{x}^H\mv{B}_1^\ast\mv{x}\end{small}&\begin{small}=\mv{x}^H\left(-\lambda^\ast\mv{H}+\sum\limits_{k=1}^K\beta_k^\ast\gamma_e\mv{G}_k+\sum\limits_{k=1}^K\alpha_k^\ast\zeta\mv{G}_k-\theta^\ast \mv{I}\right)\mv{x}\nonumber \end{small}\\ & \begin{small}=\mv{x}^H\left(\sum\limits_{k\in \Psi}\alpha_k^\ast\zeta\mv{G}_k-\theta^\ast\mv{I}\right)\mv{x}=0.\end{small} \label{eqn:negative definite}
\end{align}

To summarize, if there is no non-zero solution $\mv{x}\in \mathbb{C}^{M\times 1}$ to the equations given in (\ref{eqn:sufficient condition}), then (\ref{eqn:tight}), (\ref{eqn:null space}) and (\ref{eqn:negative definite}) cannot be satisfied at the same time, and it thus follows that the maximum eigenvalue of $\mv{D}_1^\ast$ cannot be zero, i.e., ${\rm rank}(\mv{D}_1^\ast)=M$. Then according to Proposition \ref{proposition4}, ${\rm rank}(\mv{S}^\ast)=1$ is always true for (P1.1-SDR-Eqv). Proposition \ref{proposition3} is thus proved.

\end{appendix}


\end{document}